\newif\ifanon
\anonfalse

\documentclass[11pt]{article}
\usepackage[T1]{fontenc}
\usepackage[utf8]{inputenc}
\usepackage[american]{babel}
\usepackage[letterpaper,margin=1in]{geometry}
\usepackage[hidelinks=true]{hyperref}
\usepackage{hyphenat}
\usepackage{microtype}
\usepackage{enumitem}
\usepackage{amsmath}
\usepackage{amssymb}
\usepackage{amthm}
\usepackage{local-models}
\usepackage{csquotes}
\usepackage[skins,many]{tcolorbox}
\usepackage{xcolor}
\usepackage{ambib}
\usepackage[capitalize]{cleveref}

\newcommand{\lparent}{\mathsf {P}}
\newcommand{\lleft}{\mathsf {L}}
\newcommand{\lright}{\mathsf {R}}
\newcommand{\lup}{\mathsf {U}}
\newcommand{\ldown}{\mathsf {D}}
\newcommand{\llch}{\ensuremath{\mathsf {Ch_L}}}
\newcommand{\lrch}{\ensuremath{\mathsf {Ch_R}}}

\newcommand{\lerror}{\mathsf {Error}}
\newcommand{\lgriderror}{\mathsf {GridError}}
\newcommand{\lverterror}{\mathsf {VertError}}
\newcommand{\ltreeerror}{\mathsf {TreeError}}
\newcommand{\lcolumnerror}{\mathsf {ColumnError}}

\newcommand{\lgridedge}{\mathsf {gridEdge}}
\newcommand{\lgridnode}{\mathsf {gridNode}}
\newcommand{\ltreenode}{\mathsf {treeNode}}
\newcommand{\ltreeedge}{\mathsf {treeEdge}}

\newcommand{\lgrid}{\mathsf {grid}}
\newcommand{\lvgrid}{\mathsf {vGrid}}

\newcommand{\ltreelike}{\mathsf {tree}}

\newcommand{\lpointer}{\mathsf {pointer}}
\newcommand{\lbadtree}{\mathsf {badTree}}
\newcommand{\lbadgraph}{\mathsf {badGraph}}
\newcommand{\lyes}{\mathsf {yes}}
\newcommand{\lno}{\mathsf {no}}
\newcommand{\lhard}{\mathsf {hard}}

\newcommand{\fvaluegrid}{\mathsf{val}_G}
\newcommand{\fvaluetree}{\mathsf{val}_T}
\newcommand{\inp}{\mathrm{input}}
\newcommand{\outp}{\mathrm{output}}

\usepackage{mathtools}
\DeclarePairedDelimiter{\abs}{\lvert}{\rvert}

\newtcolorbox{myframe}[2][]{%
	breakable,enhanced,colback=white,colframe=black,coltitle=black,
	sharp corners,boxrule=0.4pt,
	fonttitle=\itshape,
	attach boxed title to top left={yshift=-0.3\baselineskip-0.4pt,xshift=2mm},
	boxed title style={tile,size=minimal,left=0.5mm,right=0.5mm,
		colback=white,before upper=\strut},
	title=#2,#1
}

\newenvironment{myabstract}%
{\list{}{\listparindent 1.5em
        \itemindent    \listparindent
        \leftmargin    1cm
        \rightmargin   1cm
        \parsep        0pt}%
    \item\relax}%
{\endlist}

\newenvironment{mycover}%
{\list{}{\listparindent 0pt
        \itemindent    \listparindent
        \leftmargin    1cm
        \rightmargin   1cm
        \parsep        0pt}%
    \raggedright
    \item\relax}%
{\endlist}

\newcommand{\myaff}[1]{\,$\cdot$\, {\small #1}\par\smallskip}

\begin{document}

\begin{mycover}
{\huge\bfseries Shared Randomness Helps with \\ Local Distributed Problems \par}
\bigskip
\bigskip

\ifanon
\textbf{Anonymous authors}
\else

\textbf{Alkida Balliu}
\myaff{Gran Sasso Science Institute}

\textbf{Mohsen Ghaffari}
\myaff{MIT}

\textbf{Fabian Kuhn}
\myaff{University of Freiburg}

\textbf{Augusto Modanese}
\myaff{Aalto University}

\textbf{Dennis Olivetti}
\myaff{Gran Sasso Science Institute}

\textbf{Mikaël Rabie}
\myaff{Université Paris Cité, CNRS, IRIF}

\textbf{Jukka Suomela}
\myaff{Aalto University}

\textbf{Jara Uitto}
\myaff{Aalto University}
\fi % end ifanon

\bigskip
\end{mycover}

\begin{myabstract}
\noindent\textbf{Abstract.}
By prior work, we have many wonderful results related to distributed graph algorithms for problems that can be defined with local constraints; the formal framework used in prior work is \emph{locally checkable labeling problems} (LCLs), introduced by Naor and Stockmeyer in the 1990s. It is known, for example, that if we have a deterministic algorithm that solves an LCL in $o(\log n)$ rounds, we can speed it up to $O(\log^* n)$ rounds, and if we have a randomized algorithm that solves an LCL in $O(\log^* n)$ rounds, we can derandomize it for free.

It is also known that randomness helps with some LCL problems: there are LCL problems with randomized complexity $\Theta(\log \log n)$ and deterministic complexity $\Theta(\log n)$. However, so far there have not been any LCL problems in which the use of \emph{shared randomness} has been necessary; in all prior algorithms it has been enough that the nodes have access to their own private sources of randomness.

Could it be the case that shared randomness never helps with LCLs? Could we have a general technique that takes any distributed graph algorithm for any LCL that uses shared randomness, and turns it into an equally fast algorithm where private randomness is enough?

In this work we show that the answer is \emph{no}. We present an LCL problem $\Pi$ such that the round complexity of $\Pi$ is $\Omega(\sqrt{n})$ in the usual randomized \local model (with private randomness), but if the nodes have access to a source of shared randomness, then the complexity drops to $O(\log n)$.

As corollaries, we also resolve several other open questions related to the landscape of distributed computing in the context of LCL problems. In particular, problem $\Pi$ demonstrates that distributed \emph{quantum} algorithms for LCL problems strictly benefit from a shared quantum state. Problem $\Pi$ also gives a separation between \emph{finitely dependent distributions} and \emph{non-signaling distributions}.
\end{myabstract}

\thispagestyle{empty}
\setcounter{page}{0}
\clearpage

\section{Introduction}

In this work we present a graph problem that is solely defined with local constraints, yet distributed algorithms for solving it benefit from shared randomness.
More formally, we present a locally checkable labeling problem (LCL) $\Pi$ such
that any randomized distributed algorithm that solves $\Pi$ in the usual \local
model of distributed computing requires $\Omega(\sqrt{n})$ communication rounds,
but if we have access to shared randomness, then we can exponentially improve the
round complexity, down to $O(\log n)$ rounds.

\paragraph{\boldmath Context: LCL problems and the \local model.}

LCL problems were originally introduced by Naor and Stockmeyer
\cite{naor-stockmeyer1995} in the 1990s, and in the recent years they have
formed one of the cornerstones of the modern theory of distributed graph
algorithms. An LCL problem is simply a graph problem that can be specified by
giving a finite set of valid labeled neighborhoods. For example, the task of
coloring vertices with 10 colors in graphs of maximum degree at most 20 is an
LCL problem. (It can be defined by listing all possible radius-1 neighborhoods
of degree at most 20, and by listing for each of them all valid 10-colorings.)
Numerous problems that have been studied in distributed graph algorithms over
the decades are LCL problems (at least when restricted to bounded-degree
graphs); examples include maximal independent sets, maximal matching, various
problems related to vertex and edge coloring, and various tasks related to
orienting edges or partitioning of edges subject to local constraints. In fact 
even 3SAT can be interpreted as an LCL problem (with the bounded-degree
assumption corresponding to the case in which each variable occurs in a bounded
number of clauses).

While LCL problems are meaningful in any model of computing, they have been studied in particular from the perspective of distributed graph algorithms, and the most prominent model there is the \local model of computing \cite{linial92,peleg00distributed}. In brief, an algorithm $A$ with running time $T$ in the \local model is simply a function that maps radius-$T$ neighborhoods to local outputs. That is, to apply $A$ in a given graph $G$, each node $v$ looks at all information in its radius-$T$ neighborhood and uses $A$ to determine its own local output (for example, the color of $v$ if the task is to find a graph coloring). It turns out that we could also equivalently interpret $G$ as a computer network, and then $A$ can be interpreted as a distributed message-passing algorithm in which all nodes stop after $T$ synchronous communication rounds. We will hence interchangeably refer to $T$ as the running time, locality, or round complexity of~$A$.

A comprehensive theory of LCL problems in the \local model has been developed in
the past 10 years. 
There are numerous theorems that apply to all LCL problems, or all LCL problems in some graph family, such as trees or grids \cite{balliu18lcl-complexity,balliu20almost-global,doi:10.1137/17M1157957,dist_deran,fischer_ghaffari2017sublogarithmic,Rozhon2019,brandt16lll,chang16exponential,ghaffari17distributed,ghaffari19degree-splitting,balliu19lcl-decidability,balliu21mm,how-much-does-randomness-help-lcls,brandt17grid-lcl,balliu22rooted-trees,balliu22regular-trees,balliu22regular-trees,chang23automata-theoretic}. To give some flavor of the power of these results, here is one example: if there is a \emph{randomized} \local algorithm $A$ that solves some LCL problem $\Pi$ in $T(n) = o(\log \log n)$ rounds in $n$-node graphs, we can also construct a \emph{deterministic} \local algorithm $A'$ that solves the same problem $\Pi$ in $T'(n) = O(\log^* n)$ rounds \cite{chang16exponential}. That is, we can for free derandomize algorithms and speed them up.

In general, the relation between randomized and deterministic algorithms in the
context of LCL problems is now well understood, and we also have a clear view of
the landscape of all possible round complexities that we may have for LCL
problems \cite{suomela-2020-landscape}. However, more care is needed here:
what exactly do we mean by \emph{randomized} \local algorithms?

\paragraph{Question: shared vs.\ private randomness.}

Essentially all work on LCLs in the \randlcl model assumes that each node has its own \emph{private} source of randomness. More precisely, nodes are initially labeled independently and uniformly at random with strings of bits, and then a $T$-round algorithm can make use of all such bit strings within radius~$T$.

However, there is another notion of randomized algorithms that has been studied
for instance in the context of communication complexity: \emph{shared}
randomness (e.g.,
\cite{kurri21_coordination_ieeetit,acharya19_communication_icml}).
That is, there is one global random bit string that all nodes can see.
There are many contexts in which access to shared randomness helps
\cite{rao20_communication_book,crescenzi19_trade_disc,montealegre20_shared_isaac},
but does it help with any LCL problem?

Prior to this work, there was no evidence that shared randomness might help with LCL problems. On the contrary, all numerous LCL problems that we have encountered in prior work seem to be such that either (1) randomness does not help at all, or (2) randomness helps but private randomness is sufficient. There have even been systematic studies of infinite families of LCL problems \cite{balliu20binary-labeling,chang23automata-theoretic}, as well as computer-assisted explorations of the space of LCL problems \cite{tereshchenko21thesis}, yet there is no known candidate problem that might benefit from shared randomness. Intuitively, the key obstacle seems to be the combination that LCL problems are defined using local constraints and the set of input and output labels is finite. Shared randomness could be used to e.g.\ select a globally consistent random label from the set of finite output labels, but if that succeeds w.h.p.\ in arbitrarily large graphs, there also has to exist a deterministic choice that succeeds.

Hence, for all that we know, we might very well be living in a world in which the following conjecture is true: if an LCL problem $\Pi$ can be solved in $T(n)$ rounds with the help of shared randomness, it can also be solved in $O(T(n))$ rounds with only private randomness.

Were this to be true, it would considerably simplify the landscape of models, as
discussed further below. It would also give a helpful algorithm design
tool: we could design algorithms that exploit shared randomness, and then for
free turn them into genuine distributed algorithms that only use private
randomness. Conversely, it would allow us to strengthen all existing lower
bounds that hold for private randomness into lower bounds that extend all the
way to shared randomness.

What we show in this work is that this result cannot be true. Indeed, conversion from shared to private randomness for some LCL problems may lead to an \emph{exponential} increase in the round complexity.

\paragraph{Main contribution.}

In this work we present an LCL problem $\Pi$ such that the round complexity of $\Pi$ is $\Omega(\sqrt{n})$ in the usual randomized \local model (with private randomness), but if the nodes have access to a source of shared randomness, then the complexity drops to $O(\log n)$. This is the first known LCL that separates these two models.

Our problem $\Pi$ is an LCL exactly in the strict sense originally defined by Naor and Stockmeyer \cite{naor-stockmeyer1995}, and we do not exploit any promise on the graph family or input. Being promise-free is important, as the entire theory of LCL problems is fundamentally promise-free (for example, the known gap results would disappear if we can have arbitrary promises on the input structure), and hence also any interesting separations or counterexamples have to be promise-free.

We refer to \cref{th:private-rand,thm:ub-shared-rand} for the formal theorem statements of our lower bound and upper bound.

\paragraph{Corollary 1: distributed quantum computing.}

One of the major open questions at the intersection of distributed computing and
quantum information theory is which distributed problems admit quantum
advantage. A key model for studying this question is the \qlocal model, which is
essentially what one gets if we imagine that nodes of the input graph are
quantum computers and communication channels can be used to exchange qubits. It
is known that there are some (artificial) graph problems that can be solved in
constant time in \qlocal yet for which classical \local requires linear time
\cite{legall19_quantum_stacs}.
Nevertheless, it is still wide open whether there is any \emph{LCL} problem that
admits distributed quantum advantage; see, for instance,
\cite{akbari24_online_arxiv,coiteux-roy24quantum-coloring}.

So far, there have been two major variants of \qlocal that have been studied in
the literature: \qlocal with a shared quantum state (i.e., nodes are configured
in advance and share entangled qubits) and \qlocal without any shared quantum
state \cite{akbari24_online_arxiv,gavoille2009,arfaoui2014}.
It was not known whether either of these is (strictly) stronger than \randlcl
for any LCL problem. 
There are no known examples of LCL problems that would potentially benefit from
the shared quantum state, and it seemed reasonable to conjecture that, even if
\qlocal turns out to be is stronger than \randlcl, the shared quantum state does
not give it any additional power. 
Indeed, there were (unsuccessful) attempts at unifying the two variants of
\qlocal.

\begin{figure}[t]
	\centering
	\includegraphics[width=\textwidth]{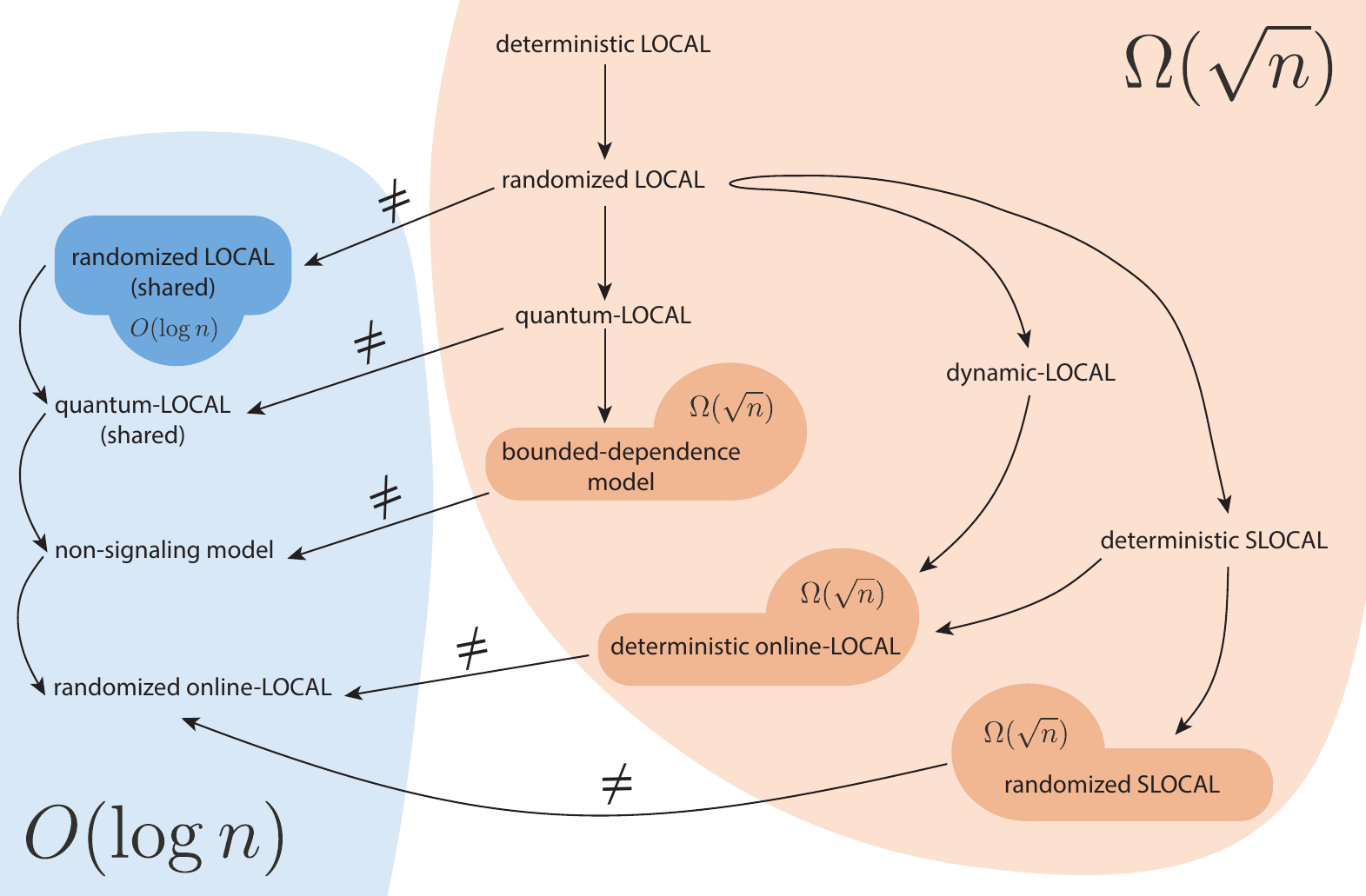}
	\caption{Landscape of models and the new separations between them. The general structure of the landscape is from \cite{akbari24_online_arxiv}. In this work we present a new LCL problem $\Pi$ that is easy in the blue-shaded region (models in which we have access to shared randomness), but hard in the red-shaded region (all other models), and we will get separations for all pairs of models that cross the cut. To prove these results, we give an upper bound in \randlcl with shared randomness (\cref{thm:ub-shared-rand}), and all upper bounds in the blue region follow, and we give lower bounds in randomized \slocal (\cref{thm:lb-slocal}), \detolcl (\cref{thm:lb-detolcl}), and the \boundep (\cref{thm:lb-boundep}), and all lower bounds in the red region follow.}
  	\label{fig:landscape}
\end{figure}

One unexpected corollary of our work is that the two variants of \qlocal are
indeed distinct, though for mundane reasons that have little to do with quantum
physics. It simply happens to be the case that a shared quantum state gives the
nodes access to shared randomness. As we show, the problem that we construct in
this work is hard in \qlocal without shared quantum state, but it becomes easy
in \qlocal with shared quantum state (since it is easy already in \randlcl with
shared randomness).

Hence, the entire question was wrong: shared quantum state does help, but for the wrong reasons. The present work highlights that the right question is whether shared quantum state provided further advantage beyond shared randomness.

\paragraph{Corollary 2: finitely dependent vs.\ non-signaling distributions.}

There is a line of research in mathematics that aims at capturing which problems
admit \emph{finitely-dependent distributions}
\cite{holroyd2016,holroyd2018,holroyd2023symmetrization,timar2024}; these are
distributions over nodes such that their restriction to a
set $X$ of nodes is independent of their restriction to another set $Y$ of
nodes if the shortest-path distance between $X$ and $Y$ is greater than some
constant. For example, the output distribution of any constant-time \randlcl
algorithm is a finitely-dependent distribution. A key question in this
context has been whether finitely-dependent distributions are strictly stronger
than constant-time \randlcl algorithms, which indeed is the case
\cite{holroyd2016}. A natural generalization of finitely-dependent distributions
to arbitrary (not necessarily constant) distance is called a
\emph{bounded-dependence distribution} \cite{akbari24_online_arxiv}.

Another closely related definition arises from quantum information theory and
the study of distributed quantum advantage: \emph{non-signaling distributions}
\cite{akbari24_online_arxiv,gavoille2009,arfaoui2014,coiteux-roy24quantum-coloring}.
Informally, a family of output distributions is non-signaling with locality $T$
if the distribution restricted to some set of nodes $X$ does not change if we
modify the input graph more than $T$ hops away from~$X$. A bounded-dependence
distribution with locality $T$ is non-signaling with locality $O(T)$, but the
converse is not necessarily true.

Prior to this work, there were no known examples of LCL problems that admit a
non-signaling distribution with locality $T$ but do not admit a
bounded-dependence distribution with locality $O(T)$. Indeed, it was again
reasonable to conjecture that no such problem exists. Our construction gives a
separation also between these two models.

\paragraph{Other corollaries.}

Our construction also gives an exponential separation between deterministic and
randomized versions of the \onlinelocal model (see
\cite{akbari24_online_arxiv,akbari23_locality_icalp}). Previously, there were no
known examples of LCLs that separate these models. For our problem we can prove
a lower bound in the \detolcl model, and the upper bound for \randlcl with
shared randomness directly works also in randomized \onlinelocal.

\paragraph{The big picture.}

All of our results are summarized in \cref{fig:landscape}. All separations between the two regions are new, in the sense that there was previously no (promise-free) LCL that would separate these pairs of models. The results that lead to this landscape are:
\begin{itemize}[noitemsep]
	\item \cref{thm:ub-shared-rand}: an upper bound in \randlcl with shared randomness,
	\item \cref{thm:lb-slocal}: a lower bound in randomized \slocal,
	\item \cref{thm:lb-detolcl}: a lower bound in \detolcl,
	\item \cref{thm:lb-boundep}: a lower bound in \boundep.
\end{itemize}
However, to keep this work easy to follow also for those who are not interested in models beyond \randlcl, we also prove the following result that is technically redundant, but serves as a warm-up for the other results:
\begin{itemize}[noitemsep]
	\item \cref{th:private-rand}: a lower bound in \randlcl without shared randomness.
\end{itemize}

\paragraph{Open questions.}

We conjecture that our problem $\Pi$ also exhibits a doubly-exponential separation between the \randlcl model and the \emph{massively parallel computing} (MPC) model \cite{KarloffSV10}. More precisely, we conjecture that our problem $\Pi$ can be solved in $O(\log \log n)$ rounds in the MPC model, while it is known to require $\Omega(\sqrt{n})$ rounds in \randlcl. Proving this is deferred for future work.

Our problem $\Pi$ fundamentally exploits the existence of short cycles (in the sense that the problem is trivial in trees and interesting only in graphs with short cycles). A key open question is \emph{whether shared randomness helps with any LCL in trees}. We have preliminary evidence suggesting that shared randomness never helps in rooted regular trees, but the case of general trees remains open.

\section{Definitions}

\paragraph{Labeled graphs.}
We start by defining the notion of labeled graph.
\begin{definition}[Labeled graph]
	Let $\mathcal{V}$ and $\mathcal{E}$ be sets of labels. A graph $G = (V,E)$ is called $(\mathcal{V},\mathcal{E})$-\emph{labeled} if:
	\begin{itemize}[noitemsep]
		\item Each node $u \in V$ is assigned a label from $\mathcal{V}$;
		\item Each node-edge pair $(u,e) \in V \times E$, satisfying $u \in e$, is
		assigned a label from $\mathcal{E}$. 
	\end{itemize} 
	A node-edge pair $(u,e)$ that satisfies $u \in e$ is also called \emph{half-edge}  incident to $u$.
\end{definition}

\begin{definition}[Labeled graph satisfying some constraints]
	Let $G$ be a graph, and let $\mathcal{C}$ be a set of constraints over the labels
	$\mathcal{V}$ and $\mathcal{E}$. The graph $G$ satisfies $\mathcal{C}$ if and
	only if:
	\begin{itemize}[noitemsep]
		\item $G$ is $(\mathcal{V},\mathcal{E})$-labeled, and
		\item the constraints of $\mathcal{C}$ are satisfied over all nodes of $G$.
	\end{itemize}
\end{definition}

\begin{definition}[Locally checkable labeleling (LCL) problem]
	A \emph{locally checkable labeling} (LCL) problem $\Pi$ is defined by a tuple
	$(\mathcal{V}_\inp, \mathcal{E}_\inp, \mathcal{V}_\outp, \mathcal{E}_\outp,
	\mathcal{C})$ where $\mathcal{C}$ is a set of constraints over
	$\mathcal{V}_\outp$ and $\mathcal{E}_\outp$.
	Given a $(\mathcal{V}_\inp,\mathcal{E}_\inp)$-labeled graph $G$, one is asked
	to label $G$ so that it satisfies $\mathcal{C}$.
\end{definition}

We denote with $L_u(e)$ the label on the half-edge $(u,e)$. Something that will
be very useful throughout the paper is to define a way to denote the node that
we can reach from a node $u$ by following some specific chain of labels assigned
to half-edges. Let $G = (V,E)$ be $(\Sigma_V,\Sigma_E)$-labeled. Let $L_1, L_2,
\ldots, L_k$ be labels in $\Sigma_E$. We define a function $f(u, L_1, L_2,
\ldots, L_k)$ that takes as input a node $u \in V$ and labels $L_1, \ldots, L_k$
in $\Sigma_E$, and returns the node $v$ reachable from $u$ by following the
unique path whose edges are labeled with $L_1, \ldots, L_k$ (in this order); if
there is no such path or it is not unique, then the value of $f$ is undefined.
More precisely, let $P = (v_1, v_2, \ldots, v_{k+1})$ be a path that starts at
$v_1=u$ and such that, for any edge $e=\{v_i,v_{i+1}\}$, the half-edge
$(v_i,e)$ is labeled with $L_{v_i}(e)=L_i$. Then 
\[
	f(u, L_1, L_2, \ldots, L_k)=
	\begin{cases}
		v_{k+1}, &\text{if $P$ exists and is unique} \\
		\bot, &\text{otherwise.}
	\end{cases}
\]

\paragraph{\boldmath The \local model.}
In the \local model of computing, we imagine that nodes of a graph $G = (V,E)$
are computers with access to unbounded computational resources (time and space).
Each computer is assigned a unique identifier from the set
$\{1,\dots,\abs{V}^c\}$, where $c \ge 1$ is some constant.
The communication between computers is as defined by $E$.
The computation proceeds in rounds where, in each round, each computer exchanges
messages of unbounded size with their neighbors and performs some local
computation (which we may perceive as instantaneous).

The above gives the deterministic variant of \local.
The randomized variant (with private randomness) is the same but where we also
give each node access to an infinite string of random bits (that is guaranteed
to be independent of the strings of other nodes in the network).
In the variant with shared randomness, nodes are given simultaneous access to
the \emph{same} infinite string.

\section{High-level ideas}
The main ingredient of our work is an LCL problem $\Pi$ with some desirable properties. On a high-level, this problem is promise-free, in the sense that it is defined on any graph. However, it is defined in such a way that there exists a family of graphs $\mathcal{G}$ that we call \emph{hard instances}:
\begin{itemize}
	\item If the graph $G$ in which the algorithm solving $\Pi$ is run is not in $\mathcal{G}$, the LCL is defined in such a way that the nodes of $G$ can produce a locally checkable proof of this fact. Moreover, such a proof can be computed ``fast''.
	\item However, if $G \in \mathcal{G}$, the LCL forces the nodes to solve a
	problem that can be solved ``fast'' if they have access to shared random bits,
	whereas any algorithm working without shared randomness must be ``slow''.
\end{itemize}
Here ``fast'' and ``slow'' depend on the precise model considered. 
For example, in the case of the \local model, the problem requires just $O(\log n)$ rounds with shared randomness, but $\Omega(\sqrt{n})$ rounds without it.

In the following, we start by providing an overview of the structure of hard
instances. Then we explain what the problem $\Pi$ is on a hard instance, and
later we will explain how the problem is defined to be promise-free.

\paragraph{The hard instances.}
Consider the graph depicted in \Cref{fig:hard-instance}. It is composed of a
square grid, where on top of each column we place a tree-like structure.
This kind of graph has a couple of useful properties:
\begin{itemize}
	\item For any two nodes in the grid belonging to the same column, it holds that their distance is $O(\log n)$.
	\item As we will see, this structure can be certified. That is, it is possible
	to provide a constant-sized certificate to the nodes such that, if the
	certificate looks good everywhere, then the graph is indeed a hard instance.
	Conversely, if the graph is not a hard instance, then any assignment of the
	certificate leads to an error somewhere. We will make use of this fact later
	when making the problem promise-free.
\end{itemize}
\begin{figure}
	\centering
	\includegraphics[width=\textwidth]{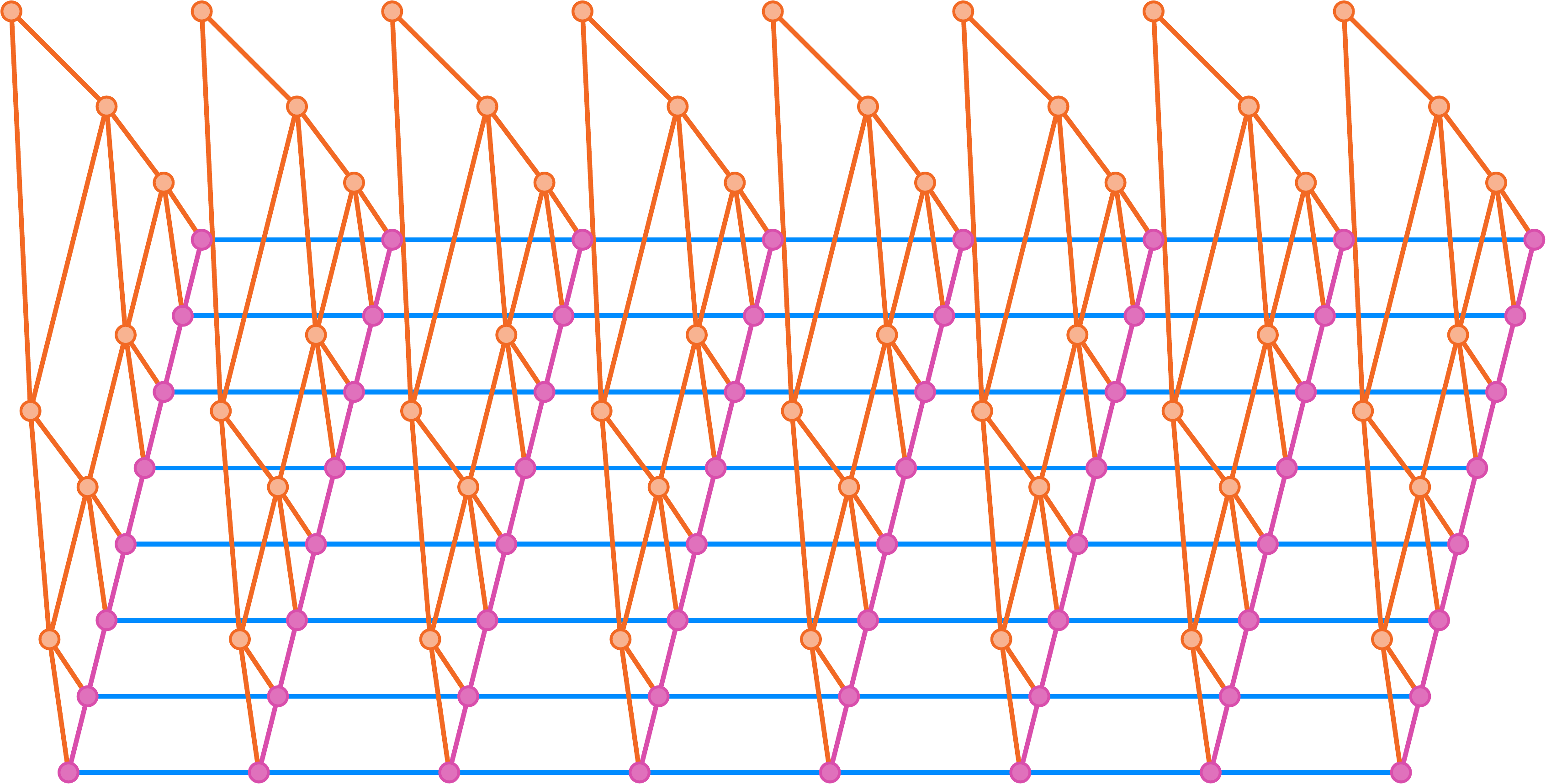}
	\caption{An example of a hard instance. The grid is composed of blue edges, purple edges, and purple nodes. Each connected component induced by orange nodes, orange edges, purple edges, and purple nodes connected to the orange edges is a tree-like structure.}\label{fig:hard-instance}
\end{figure}

\paragraph{\boldmath The LCL problem $\Pi^{\lhard}$ defined on hard instances.}
Consider the following problem $\Pi^{\lhard}$, defined on hard instances. Each
node in the right-most column of the grid receives a bit as input. Then, all
nodes of the grid must output a bit such that the two following constraints are
satisfied:
\begin{enumerate}
	\item[C1:] For each row of the grid, all nodes must output the same bit.
	\item[C2:] There must exist at least one row such that the output bit is the
	same as the input bit of the right-most node of that row.
\end{enumerate} 
It is possible to present this problem as an LCL as follows:
\begin{itemize}
	\item Assume that the grid has an input labeling encoding its orientation;
	that is, each node knows which of its neighbors is on its left, right,
	above, and below. Then constraint C1 can be encoded in the LCL by requiring
	every node to have the same output as its left and right neighbors.
	\item In order to enforce constraint C2, we use the tree-like structure on top
	of the last column. Say a grid node of the last column is \emph{happy} if its
	output agrees with its input. Meanwhile an inner node of the tree is happy if
	at least one of their children is happy. All nodes of the tree output whether
	they are happy, and we require the root to be happy. These constraints are
	local, and in fact they can be encoded as an LCL.
\end{itemize}
An example of a valid output is shown in \Cref{fig:solution}.

\begin{figure}[t]
	\centering
	\includegraphics[width=0.6\textwidth]{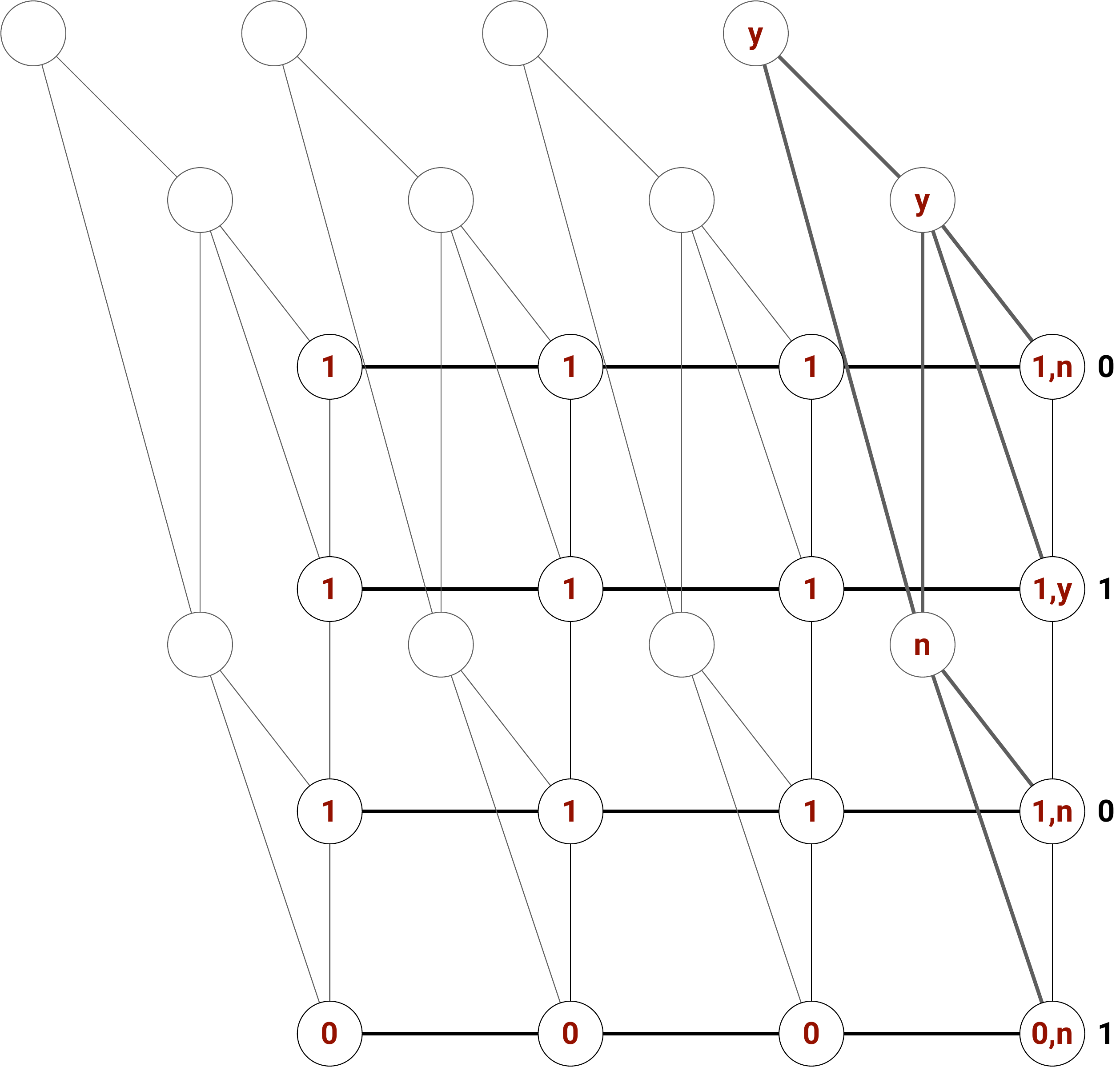}
	\caption{An example of a solution for $\Pi^{\lhard}$. Black bits represent the inputs of the nodes of the last column. The inputs of the other nodes do not affect the solution and are omitted. Labels in red represent the outputs, where the label $y$ represents a happy node, and the label $n$ represents an unhappy node. All nodes that are not labeled either $y$ or $n$ output $y$, which is omitted in the figure.}\label{fig:solution}
\end{figure}

Now, given such a problem, we can obtain a separation between shared and private
randomness in the \local model: 
\begin{itemize}
	\item If the nodes have access to shared randomness, then each node only needs
	$O(\log n)$ rounds to determine its vertical position $y$ in the grid.
	Once it has figured out the value of $y$, the node simply outputs the $y$th
	shared random bit.
	In this way, all nodes in the same row output the same bit, and for each row
	with probability $1/2$ we have that this bit agrees with the input given to
	the node on the right-most column. 
	Since there are $\Theta(\sqrt{n})$ rows, we get that there is at least one
	good row with high probability.
	\item Conversely, if only private randomness is allowed, nodes in the same row
	do not have any way to coordinate their output and communication across the
	whole row is expensive ($\Omega(\sqrt{n})$ rounds).
	The only alternative is for nodes in the same row to deterministically fix
	their output as a function of their vertical position.
	In this case we can adversarially pick the input to the rightmost node so that
	the row does not succeed.
	Since this cannot hold for every row (as otherwise the algorithm would not be
	solving $\Pi$), we get the lower bound of $\Omega(\sqrt{n})$ rounds.
\end{itemize}

Although the idea of the problem is simple, the main challenge is making the
problem promise-free. 
That is, we also have to account for all the cases where the input graph is not
as in \Cref{fig:hard-instance}.
Next we provide an overview of this process.

\subsection{A tree-like structure}\label{ssec:roadmap:treelike}
A useful property that our problem satisfies is the following: nodes
that belong to the same column of the grid should either be able to see the whole
column by inspecting their $O(\log n)$-radius neighborhood, or the nodes can
\emph{prove that there is some error} within distance $O(\log n)$. In
\Cref{sec:treelike}, we describe \emph{tree-like structures} and how they give
us exactly this property. This kind of structure has already been used in
\cite{congest-lcls}; in fact some useful properties that we exploit here have
already been proved in \cite{congest-lcls}. 

\paragraph{Definition of the tree-like structure.}
Informally, a tree-like
structure is a perfect binary tree in which nodes at the same depth are also
connected via a path. Such a structure can be certified by assigning a label to
each node-edge pair.
An example of this structure as well as an assignment of a certificate for it
are depicted in \Cref{fig:treelike-labeled}. For example, the
certificate ensures that all leaves are at the same depth by requiring
that, starting from a node not having any incident edge labeled $\llch$ or
$\lrch$ (i.e., it does not have any children), and following the edge labeled
$\lright$, we must reach a node that also does not have any children.

\begin{figure}
	\centering
	\includegraphics[width=0.7\textwidth]{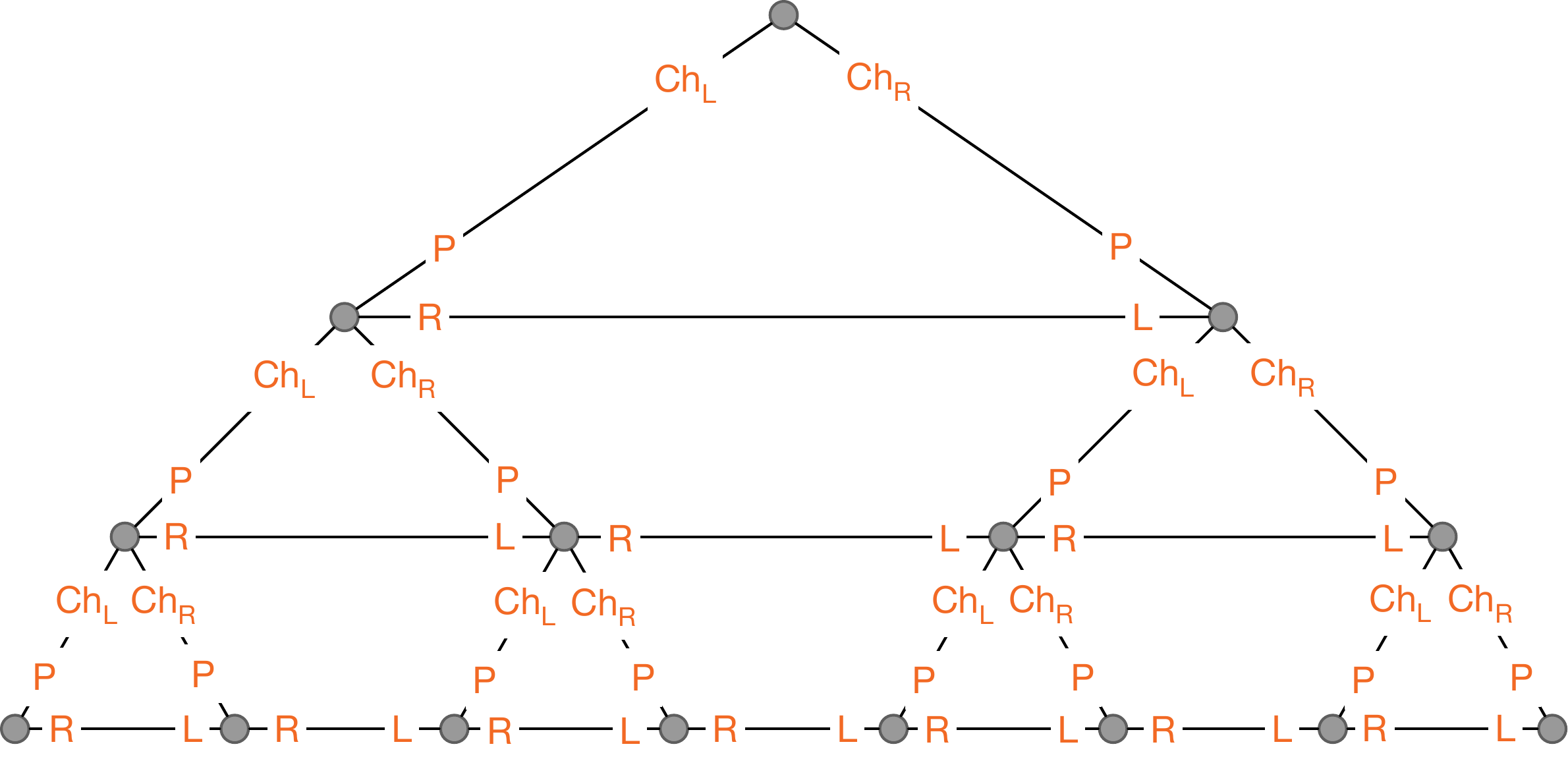}
	\caption{An example of a properly labeled tree-like structure.
	The labels $\lleft$, $\lright$, $\lparent$, $\llch$, and $\lrch$, stand, respectively, for left, right, parent, left child, and right child.
	}\label{fig:treelike-labeled}
\end{figure}

\paragraph{Local certification of tree-like structures.}
In \Cref{sec:treelike}, we start by describing tree-like structures from a
local perspective. In particular, we define a set of input half-edge labels
$\mathcal{E}^{\ltreelike}$ of constant size and a set of constraints
$\mathcal{C}^{\ltreelike}$ over constant distance that satisfy the following:
\begin{itemize}
	\item For all tree-like structures $G$, there exists an assignment of labels of $\mathcal{E}^{\ltreelike}$ to the half-edges of the graph such that the constraints of $\mathcal{C}^{\ltreelike}$ are satisfied on all nodes of $G$.
	\item Let $G$ be a graph where half-edges are labeled with labels from
	$\mathcal{E}^{\ltreelike}$ and such that the constraints of
	$\mathcal{C}^{\ltreelike}$ are satisfied on all nodes of $G$. Then, $G$ is a
	tree-like structure.
\end{itemize}
In other words, we show that there exists a locally checkable proof of constant size for the fact that the graph is tree-like.

\paragraph{The tree-like structure as an LCL.}
Building on the previously defined locally checkable proof, we define an LCL problem $\Pi^{\lbadtree}$ that satisfies the following:
\begin{itemize}
	\item For all tree-like graphs $G$, there exists an input for $\Pi^{\lbadtree}$ that can be assigned to $G$ such that the only valid solution for $\Pi^{\lbadtree}$ is the one assigning $\bot$ to all nodes of $G$.
	\item Let $G$ be a graph where half-edges are labeled with labels from $\mathcal{E}^{\ltreelike}$ such that the constraints of $\mathcal{C}^{\ltreelike}$ are not satisfied on at least one node of $G$. Then there exists a solution for $\Pi^{\lbadtree}$ where all nodes produce an output different from $\bot$. Moreover, such a solution can be computed in $O(\log n)$ rounds in the \local model. 
\end{itemize}
In other words, we define an LCL problem where the inputs are from
$\mathcal{E}^{\ltreelike}$ and where the nodes have two options: they can either
prove that the graph is not tree-like, or they can do nothing (by outputting
$\bot$). This problem is defined in such a way that an output different from
$\bot$ can be used only on structures that are not tree-like (or on tree-like
structures whose input labels are incorrect) whereas, if an output different
from $\bot$ can be used, then this can be done relatively fast ($O(\log
n)$ rounds in the \local model).

\paragraph{\boldmath How we will use $\Pi^{\lbadtree}$.}
We now describe how the problem $\Pi^{\lbadtree}$ is used when defining our main
LCL problem $\Pi$. The problem $\Pi$ is defined in such a way that all nodes
receive an input indicating whether they are part of a grid, or whether they are
part of a tree-like structure. Then $\Pi$ is defined so that:
\begin{itemize}
	\item Nodes can \emph{mark} areas of the graph that do not look like valid hard instances. In particular, nodes having an input indicating that they are part of a tree-like structure are considered marked if they solve $\Pi^{\lbadtree}$ by giving an output different from $\bot$.
	\item We prove that nodes can efficiently mark bad parts of the graph
	such that the remaining connected components \emph{almost} look like hard
	instances. Although these remaining parts do not exactly look like hard
	instances (we provide more details about this later), there is an efficient
	algorithm for solving them.
\end{itemize}

\subsection{A grid structure}\label{ssec:roadmap:grid}
As already mentioned, our hard instances are grids in which we connect a
tree-like structure on top of each column. In \Cref{sec:grid}, we describe a
structure that we call \emph{grid structure}.
A similar structure has already been used in \cite{congest-lcls} and in fact
some useful properties follow directly from what was proved in that paper.

\paragraph{Definition of the grid structure.}
Informally, a grid structure denotes a two-dimensional grid that does not
wrap around (i.e., it is not a torus). Unlike the case of tree-like
structures, we cannot achieve full certification of grids; in particular our
scheme for certifying grids is defective and also allows one to label invalid
grid structures (and in particular tori) in a way that no node sees any error. 

\Cref{fig:grid-labeled} illustrates a grid structure and its corresponding
certificate.
An example of a constraint that needs to be checked is that, if a node has
$\lright$ (i.e., ``right'') on an incident half-edge, then the corresponding
half-edge must be labeled $\lleft$ (i.e., ``left'').

\begin{figure}
	\centering
	\includegraphics[width=0.4\textwidth]{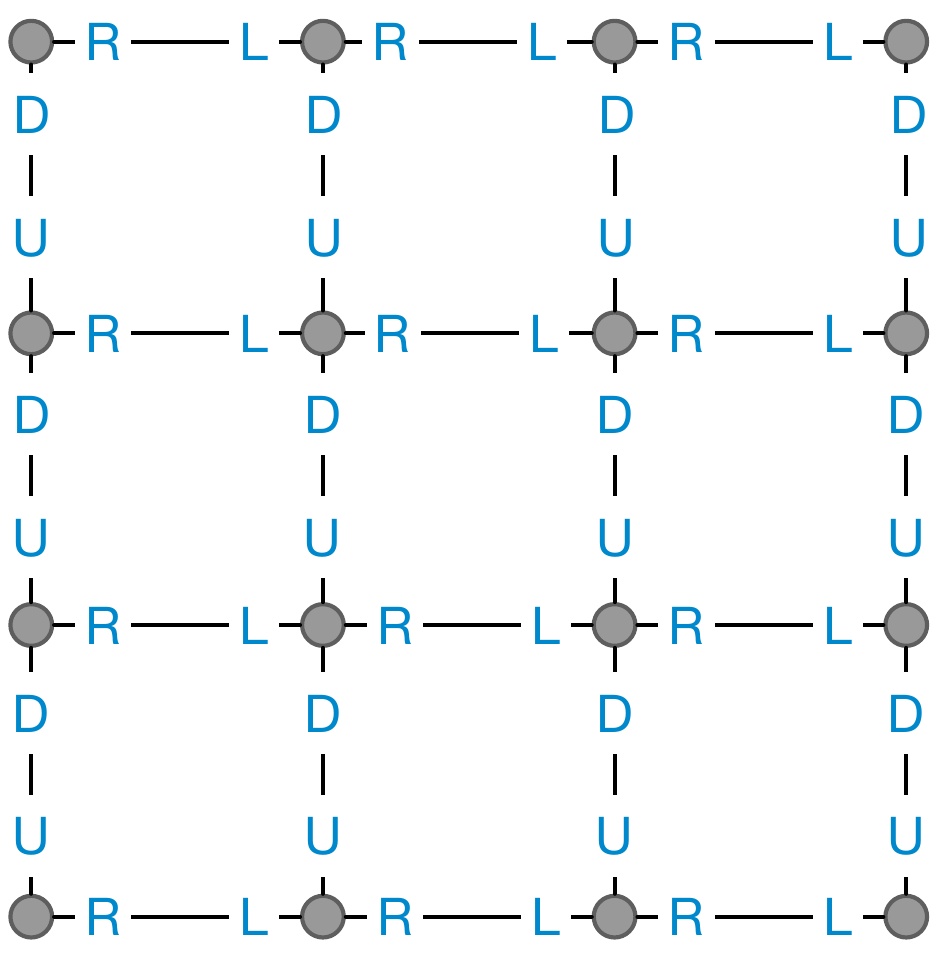}
	\caption{An example of a properly labeled grid structure. The labels $\lleft$, $\lright$, $\lup$, and $\ldown$ stand, respectively, for left, right, up, and down.}\label{fig:grid-labeled}
\end{figure}

\paragraph{Local certification of grid structures.}
In \Cref{sec:grid}, we describe grid structures from a local perspective.
In particular, we define a set of input half-edge labels $\mathcal{E}^{\lgrid}$
of constant size as well as a set of constraints $\mathcal{C}^{\lgrid}$ over constant
distance that satisfy the following:
\begin{itemize}
	\item For any grid structure $G$, there exists an assignment of labels of $\mathcal{E}^{\lgrid}$ to the half-edges of the graph such that the constraints of $\mathcal{C}^{\lgrid}$ are satisfied on all nodes of $G$.
	\item Let $G$ be a graph where half-edges are labeled with labels from
	$\mathcal{E}^{\lgrid}$ and such that the constraints of $\mathcal{C}^{\lgrid}$
	are satisfied on all nodes of $G$. 
	Moreover, suppose that there exists at least one node that has no incident half-edge labeled $\ldown$ (or $\lup$) and that there exists at least one node that has no incident half-edge labeled $\lleft$ (or $\lright$).
	Then $G$ is a grid structure.
\end{itemize}
In other words, if we assume that all nodes satisfy the constraints of $\mathcal{C}^{\lgrid}$ and we additionally assume that there is at least one node satisfying some additional constraints (which essentially guarantee that there is some ``corner'' of the grid, thus preventing the graph from being a torus), then the graph is indeed a grid.

\paragraph{Enforcing the dimensions of the grid.}
As previously discussed, when defining our main problem $\Pi$, we allow nodes to
\emph{mark} areas of the graph that do not look like valid hard instances.
Recall we cannot guarantee that unmarked areas are exactly hard instances.
Nevertheless, we are able to ensure that the unmarked parts of the graphs are
grids (with properly attached tree-like structures) that are at least as tall as
they are large. We call such grids \emph{vertical}. We later discuss how this
property is sufficient to obtain an efficient algorithm that has access to
shared randomness.

On a high-level, in order to enforce a grid to be vertical, we define a set of input node labels $\mathcal{V}^{\lvgrid}$, and a set of constraints $\mathcal{C}^{\lvgrid}$, that satisfy the following.
\begin{itemize}
	\item For any vertical grid structure $G$, there exists an assignment of labels of $\mathcal{E}^{\lgrid}$ to the half-edges of $G$ and an assignment of labels of $\mathcal{V}^{\lvgrid}$ to the nodes of $G$, such that the constraints of $\mathcal{C}^{\lgrid}$ and $\mathcal{C}^{\lvgrid}$  are satisfied on all nodes of $G$.
	\item Suppose $G$ is a graph where half-edges are labeled with labels from $\mathcal{E}^{\lgrid}$, nodes are labeled with labels from $\mathcal{V}^{\lvgrid}$, such that  the constraints of $\mathcal{C}^{\lgrid}$ and $\mathcal{C}^{\lvgrid}$  are satisfied on all nodes of $G$. Then, $G$ is a vertical grid structure.
\end{itemize} 
We note that, while the second point provides the intuition of what we will do, the statement, as is, is false. We will provide more details in \Cref{sec:grid}.

\paragraph{How we will use vertical grids.}
Consider a hard instance in which the grid is more large than tall, and in particular consider the extreme case in which the grid is just a path of linear length. In this case, in the case of shared randomness, if we apply the \local algorithm for solving $\Pi^{\lhard}$ described at the beginning of the section, we would have a success probability of just $1/2$, since there is only a single row in the grid. In order to guarantee a large-enough success probability, the problem $\Pi$ will be defined such that unmarked regions are \emph{vertical} grid structures (possibly of small size). This will guarantee the following.
\begin{itemize}
	\item If the grid has height less than $\log n$, then its width is also less than $\log n$, which implies that nodes can see the whole grid in just $O(\log n)$ rounds, and solve the problem $\Pi^{\lhard}$ by brute force.
	\item If the grid has height strictly larger than $\log n$, then the success probability will be at least $1 - 1/2^{\log n} = 1 - 1/n$, and hence the algorithm will succeed with high probability.
\end{itemize}
By combining the above, we will get that $O(\log n)$ rounds will be an upper bound on the runtime for succeeding in solving $\Pi^{\lhard}$ with high probability when using shared randomness.

\subsection{Family of hard instances}

In \Cref{sec:graph-family}, we formally define the family $\mathcal{G}$ of hard instances. These graphs are similar to the ones informally explained at the beginning of this section (see \Cref{fig:hard-instance} for an example), with the only difference that grids do not need to be squares, but they only need to satisfy that their height is at least as large as their width. Then, we define an LCL $\Pi^{\lbadgraph}$ satisfying the following.
\begin{itemize}
	\item There are two possible types of output, and different nodes could give outputs of different type. 
	\item One possible output is the empty output, and if a graph $G$ is in $\mathcal{G}$, the problem $\Pi^{\lbadgraph}$ is defined such that all nodes must produce the empty output.
	\item The other possible output is a proof for the fact that $G \notin \mathcal{G}$. More in detail,  if the graph is not in the family, nodes can spend $O(\log n)$ time in the \local model to produce a proof of this fact, such that, the subgraph $G'$ induced by nodes producing an empty output satisfies that each connected component of $G'$ is a graph in $\mathcal{G}$.
\end{itemize}
Informally, our main problem $\Pi$ will be defined such that all nodes need to solve $\Pi^{\lbadgraph}$, and then, on the subgraph  induced by nodes producing an empty output for $\Pi^{\lbadgraph}$, nodes need to solve the problem  $\Pi^{\lhard}$ that we informally defined on hard instances at the beginning of this section. The definition of $\Pi$ will satisfy that, if we consider some graph $G \in \mathcal{G}$, the only possible way to solve $\Pi$ is by producing an empty solution for $\Pi^{\lbadgraph}$, implying that nodes must then solve  $\Pi^{\lhard}$ on the whole graph. On the other hand, if a graph $G$ is not in $\mathcal{G}$, nodes can quickly (i.e., in $O(\log n)$ rounds in the \local model) mark bad parts of the graph, and then solve $\Pi^{\lhard}$ on the connected components that are part of $\mathcal{G}$.
In other words, the problem $\Pi^{\lbadgraph}$ is the one allowing us to remove the promise in the definition of $\Pi^{\lhard}$.

\paragraph{\boldmath High-level ideas behind the definition of $\Pi^{\lbadgraph}$.}
Observe that, by how hard-instances are constructed, nodes can either be exclusively part of a tree-like structure, or they can belong at the same time to the grid and to some tree-like structure. However, edges can be of three types: they can be exclusively part of the grid, exclusively part of a tree-like structure, or be part of both. In the problem $\Pi^{\lbadgraph}$, nodes and edges are input labeled to indicate to which structure(s) they belong to. Then, the possible outputs for $\Pi^{\lbadgraph}$ are the following.
\begin{itemize}
	\item If on a node the constraints of $\mathcal{C}^{\lvgrid}$ or the constraints of $\mathcal{C}^{\ltreelike}$ are not locally satisfied, or the input labeling of $\Pi^{\lbadgraph}$ is such that there is some local error in how the two structures are connected, then the node can output an error.
	\item Consider the subgraph obtained by excluding edges labeled as horizontal edges (i.e., $\lleft$ and $\lright$) of the grid. Each connected component is either a valid grid column with a properly-attached tree-like structure on top, or not. In the latter case, we also include the scenario in which some node in the connected component gave error in the previous step. In each connected component, nodes need to solve $\Pi^\lbadtree$, and thus we get the following two cases:
	\begin{itemize}
		\item The connected component looks good (i.e., it contains no nodes that output error), and the only solution for $\Pi^\lbadtree$ is the one giving $\bot$ (i.e., an empty output) on all nodes;
		\item The connected component does not look good, and nodes can (efficiently) solve $\Pi^\lbadtree$ such that no node uses the output $\bot$.
	\end{itemize}
\end{itemize}
An example of an input labeling for $\Pi^{\lbadgraph}$ that forces all nodes to output $\bot$ is depicted in \Cref{fig:hard-instance-labeled}.

\begin{figure}[t]
	\centering
	\includegraphics[width=0.6\textwidth]{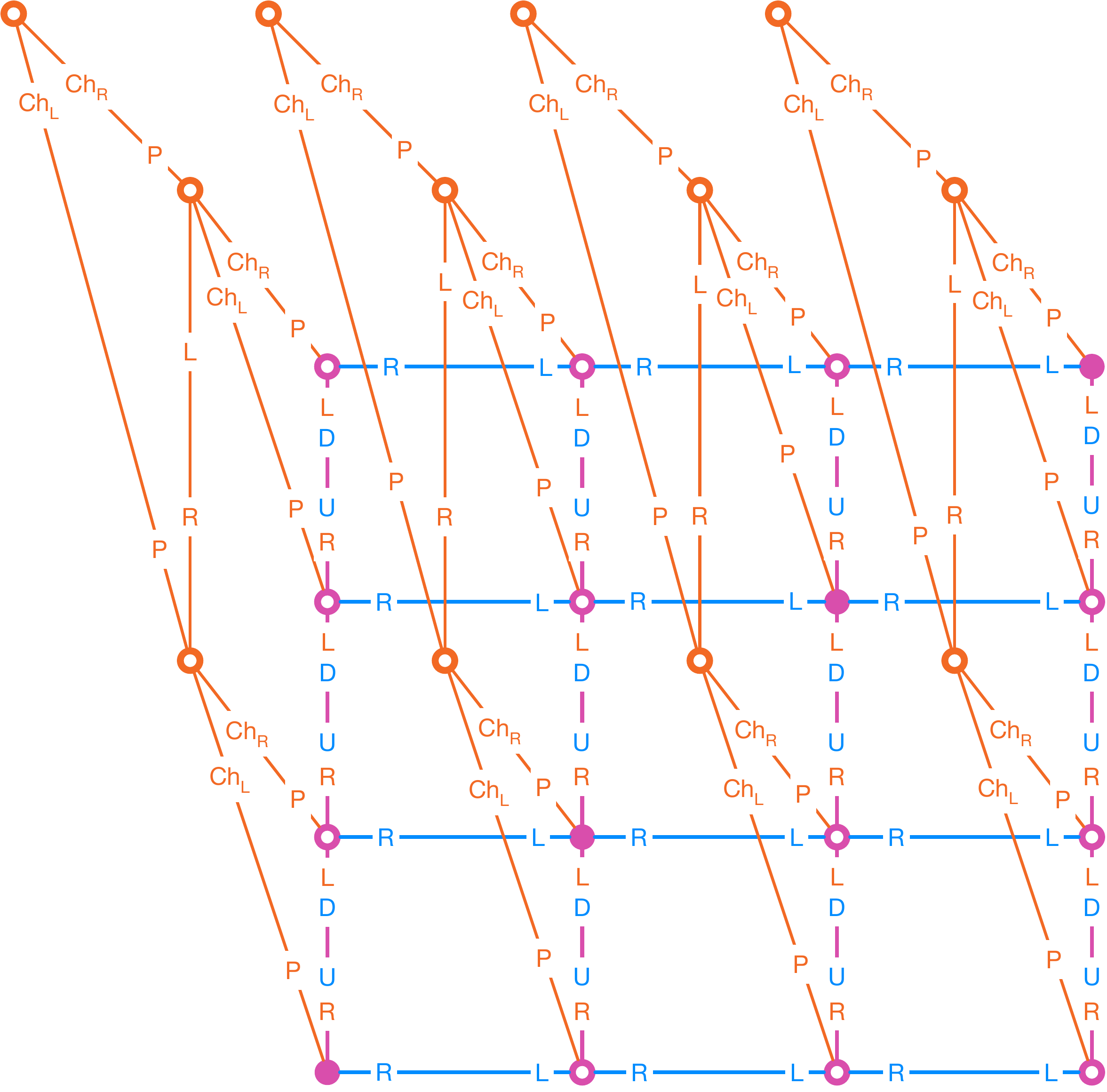}
	\caption{An example of a properly input-labeled hard instance. Orange nodes are labeled $(\ltreenode)$, purple empty nodes are labeled $(\ltreenode,\lgridnode,0)$, purple full nodes are labeled $(\ltreenode,\lgridnode,1)$, orange half-edges with label $\ell$ are labeled $(\ltreeedge,\ell)$, blue half-edges with label $\ell$ are labeled $(\lgridedge,\ell)$, purple half-edges labeled $\ell$ in blue and $\ell'$ in orange are labeled $((\lgridedge,\ell),(\ltreeedge,\ell'))$}.
	\label{fig:hard-instance-labeled}
\end{figure}

\paragraph{\boldmath Properties of $\Pi^{\lbadgraph}$.}
In \Cref{sec:graph-family}, we will prove that the following properties are satisfied by our construction:
\begin{itemize}
	\item For any $G \in \mathcal{G}$, i.e., for any hard instance, it is possible to assign an input labeling on $G$, such that the only valid output for $\Pi^{\lbadgraph}$ is the one where each node outputs $\bot$ (i.e., an empty output).
	\item For any graph $G$, there exists a solution for $\Pi^{\lbadgraph}$, which can be computed in $O(\log n)$ rounds in the \local model, such that the subgraph induced by nodes that output $\bot$ satisfies that each connected component is a graph in $\mathcal{G}$.
\end{itemize}
Note that, as discussed earlier, these two properties are exactly the ones that we need in order to make  $\Pi^{\lhard}$ promise-free.

\subsection{Problem \texorpdfstring{\boldmath$\Pi$}{Pi} and its complexity in the \local model}

We already discussed the high-level idea behind the definition of $\Pi$, and in \Cref{sec:problem-pi} we  define the LCL problem $\Pi$ more formally. We now give a bit more details about $\Pi$. 
Each node $u$ receives a pair of inputs. One input is exactly the input for $\Pi^{\lbadgraph}$, and the other input is a single bit $b_u$. The problem $\Pi$ is defined such that each node needs to solve $\Pi^{\lbadgraph}$, and then, on the subgraph  induced by nodes producing an empty output for $\Pi^{\lbadgraph}$, each node $u$ needs to solve the problem $\Pi^{\lhard}$, where the input of node $u$ for $\Pi^{\lhard}$ is $b_u$. While $\Pi^{\lhard}$ was explained on square grids of size $\Theta(\sqrt{n}) \times \Theta(\sqrt{n})$, as we already discussed, hard instances can be a bit different from square grids, which makes it harder to prove an upper bound for the problem $\Pi$.
In \Cref{sec:pi-complexity} we will provide lower and upper bounds for the problem $\Pi$ in different models of computation. 

\paragraph{\boldmath The complexity of $\Pi$ in the \local model with shared randomness.}
When we discussed vertical grids, we also provided the high-level idea of the upper bound for solving $\Pi^{\lhard}$ with shared randomness in the \local model on hard instances, which we now recap and explain how it is used to solve $\Pi$ in any graph. At first, the nodes spend $O(\log n)$ rounds to mark the ``bad'' parts of the graph. After that, each remaining connected component is a hard instance. Then, the problem $\Pi^{\lhard}$ is solved by brute force on components of diameter $O(\log n)$, while shared randomness is used to solve components of larger diameter.

\paragraph{\boldmath The complexity of $\Pi$ in the \local model with private randomness.}
In order to prove a lower bound of $\Omega(\sqrt{n})$ in the \local model for private randomness, in \Cref{sec:pi-complexity}, we essentially show that any algorithm that is too fast cannot deviate too much from being deterministic, or in other words, for each row of a hard instance, the algorithm must fix the output bits almost deterministically (i.e., the output must always be the same, with high probability). Then, we fix the input of the last column in an adversarial way, as a function of the almost-deterministic outputs of the algorithm, and we prove that in such case the failure probability of the algorithm is too large.

\subsection{Complexity of \texorpdfstring{\boldmath$\Pi$}{Pi} in other models}

In \cref{sec:other-models} we extend the $\Omega(\sqrt n)$ lower bound to
several other models that are far more powerful than \local:
\begin{itemize}
	\item \textbf{\slocal with private randomness (\cref{sec:slocal}).}
	In the \slocal model \cite{ghaffari17_complexity_stoc}, nodes are revealed in
	a sequential, potentially adversarial order.
	Each time a node is revealed, it sees all the states of previously revealed
	nodes that are at most $T$ hops away from it (in particular also their
	outputs) and is asked to commit to an output.
	This model is clearly more powerful than \local since the sequential
	processing of nodes gives us symmetry breaking essentially for free.
	\item \textbf{\Detolcl model (\cref{sec:detolcl}).}
	The \detolcl model \cite{akbari23_locality_icalp} is similar to \slocal
	in that the nodes are also revealed following some sequential order.
	Nevertheless, it is potentially more powerful than \slocal because it is able
	to maintain \emph{global} knowledge of what has been revealed thus far. 
	We show the lower bound for the deterministic variant of \onlinelocal.
	A randomized variant of \onlinelocal also exists
	\cite{akbari24_online_arxiv}, but clearly it has access to shared
	randomness by definition, and so the upper bound from \local extends there.
	\item \textbf{\Boundep (\cref{sec:boundep}).}
	The \boundep encompasses all algorithms satisfying the property that the
	output of any node is independent of outputs that are more than $T$ hops away
	from it.
	This includes, for instance, all \qlocal algorithms without a pre-shared
	quantum state.
	The \boundep has been shown to be less powerful than the randomized version of
	\onlinelocal \cite{akbari24_online_arxiv}, but its relation to
	\detolcl is still unclear.
\end{itemize}

Throughout this work we assume that the nodes have (implicitly or explicitly) some knowledge on the total number of nodes $n$. In \cref{appx:monte-carlo} we show that this is a necessary assumption.

\section{A tree-like structure}\label{sec:treelike}
In this section, we first formally define what is a tree-like structure, then we provide some local constraints that, if satisfied on all nodes, guarantee that the graph is tree-like, and finally we define an LCL problem allowing nodes to quickly prove that the graph is not tree-like.

\begin{definition}[Tree-like structure]
	We say that a graph $G$ is a \emph{tree-like structure} of height $\ell$ if it is possible to assign coordinates $(l_u, k_u)$ to each node $u\in G$, where
	\begin{itemize}[noitemsep]
		\item $0\le l_u < \ell$ denotes the depth of $u$ in the tree, and
		\item $0\le k_u < 2^{l_u}$ denotes the position of $u$ (according to some order) in layer $l_u$,
	\end{itemize}
	such that there is an edge connecting two nodes $u,v\in G$ with coordinates $(l_u, k_u)$ and $(l_v, k_v)$ if and only if:
	\begin{itemize}[noitemsep]
		\item $l_u = l_v$ and $|k_u - k_v | = 1$, or
		\item $l_v = l_u-1$ and $k_v = \lfloor \frac{k_u}{2} \rfloor$, or
		\item $l_u = l_v-1$ and $k_u = \lfloor \frac{k_v}{2} \rfloor$.
	\end{itemize}
\end{definition}

\subsection{Local checkability of a tree-like structure}
We now define sets of labels $\mathcal{V^\ltreelike}$ and $\mathcal{E^\ltreelike}$, and a set of local constraints over such labels, satisfying that a graph $G$ can be $(\mathcal{V^\ltreelike},\mathcal{E^\ltreelike})$-labeled satisfying the constraints if and only if $G$ is a tree-like structure.
The tree-like structure that we use in this paper is the same exact one used in \cite{congest-lcls}, and hence we now report here the constraints as defined in \cite{congest-lcls}.

First of all, nodes are not labeled at all, or equivalently, all nodes have the same label $\bot$. Hence, we define $\mathcal{V^\ltreelike} = \{\bot\}$. Then, the possible half-edge labels are $\mathcal{E^\ltreelike} = \{\lleft,\lright, \lparent,\llch,\lrch\}$ (the labels stand for ``left'', ``right'', ``parent'',  ``left child'', and ``right child'', respectively). 
The set of local constraints $\mathcal{C}^{\ltreelike}$ is defined as follows.

\begin{myframe}{The constraints $\mathcal{C}^{\ltreelike}$ of \cite{congest-lcls}}
	\begin{enumerate}
		\item For any two edges $e,e'$ incident to a node $u$, it must hold that $L_u(e)\neq L_u(e')$;\label{cons-tree:differentEdgeLabels}
		
		\item For each edge $e=\{u,v\}$, if $L_u(e)=\lleft$, then $L_v(e)=\lright$, and vice versa; \label{cons-tree:left-right}
		
		\item For each edge $e=\{u,v\}$, if $L_u(e)=\lparent$, then $L_v(e)\in\{\llch,\lrch\}$, and vice versa;\label{cons-tree:parent-child}
		
		\item If a node $u$ has an incident edge $e=\{u,v\}$ with label $L_u(e)=\lparent$ such that $L_v(e)=\llch$, then $f(u, \lparent,\lrch,\lleft)=u$;\label{cons-tree:triangle}
		
		\item If a node $u$ has an incident edge $e=\{u,v\}$ with label $L_u(e)=\lparent$ such that $L_v(e)=\lrch$, if $u$ has an incident edge labeled $\lright$, then $f(u, \lparent,\lright,\llch,\lleft)=u$.\label{cons-tree:square}
		
		\item If a node has an incident half-edge labeled $\llch$, then it must also have an incident half-edge labeled $\lrch$, and vice versa;\label{cons-tree:2children}
		
		\item A node does not have an incident half-edge labeled $\lparent$ if and only if it has no incident half-edges labeled $\lleft$ or $\lright$; \label{cons-tree:root}
		
		\item If a node $u$ does not have an incident edge $e$ with label $L_u(e)\in\{\llch, \lrch\}$, then neither do nodes $f(u,\lleft)$ and $f(u,\lright)$ (if they exist);\label{cons-tree:boundarychildren}
		
		\item If a node $u$ has an incident edge $e=\{u,v\}$ with label $L_u(e)=\lparent$ such that $L_v(e)=\lrch$ (resp. $L_v(e)=\llch$), then $u$ has an incident edge labeled $\lright$ (resp. $\lleft$) if and only if $f(u,\lparent)$ has an incident edge labeled $\lright$ (resp. $\lleft$).\label{cons-tree:boundarylr}
		
	\end{enumerate}
\end{myframe}
\noindent An example of valid tree-like structure labeled according to $\mathcal{C}^{\ltreelike}$ is depicted in \Cref{fig:treelike-labeled}.

In \cite{congest-lcls}, it has been shown that a graph $G$ can be $(\mathcal{V^\ltreelike},\mathcal{E^\ltreelike})$-labeled satisfying $\mathcal{C}^{\ltreelike}$ if and only if it is a tree-like structure. This result is captured by the following two lemmas.
\begin{lemma}[\cite{congest-lcls}]\label{lem:treelike}
	Let $G$ be a graph that is labeled with labels in $\mathcal{E^\ltreelike}$ such that $\mathcal{C}^{\ltreelike}$ is satisfied for all nodes in $G$. Then, $G$ is a tree-like structure.
\end{lemma}

\begin{lemma}[\cite{congest-lcls}]\label{lem:treelike2}
	Each tree-like structure graph $G$ can be labeled with labels in $\mathcal{E^\ltreelike}$ such that the constraints $\mathcal{C}^{\ltreelike}$ are satisfied at all nodes in $G$. 
\end{lemma}

\subsection{Proving that a graph is not tree-like}
We now define an LCL problem $\Pi^{\lbadtree}$, already informally introduced in \Cref{ssec:roadmap:treelike}. We omitted some details there, that we now present.
While we described tree-like structures as separate entities, we will use such structures in connection to another one, i.e., a grid. For this reason, we want to allow nodes to not only prove that the tree-like structure is invalid, but also to produce a proof in the case in which the tree-like structure is valid but wrongly connected to the rest of the graph. In order to capture this, in the definition of $\Pi^{\lbadtree}$, nodes receive some input, and in particular nodes could be marked, indicating that there is some issue unrelated to the tree-like structure per se. More in detail, the high-level idea behind the definition of $\Pi^{\lbadtree}$ is the following:
\begin{itemize}
	\item Nodes receive an input indicating whether they are marked or not;
	\item There are two possible types of output, either nodes give an empty output, or nodes produce a proof that the graph is not tree-like or that it contains at least one marked node;
	\item If the graph is tree-like and does not contain any marked node, the only valid solution is the one where all nodes produce an empty output.
	\item If the graph is not tree-like, or it contains at least one marked node, nodes can spend $O(\log n)$ time in the \local model to produce a proof of this fact.
\end{itemize}

More formally, the problem $\Pi^{\lbadtree}$ is defined on $(\{0,1\},\mathcal{E^\ltreelike})$-labeled graphs, that is, it is assumed that every half-edge has an input label in $\mathcal{E^\ltreelike}$, and that each node is either labeled $0$ or $1$, where nodes labeled $1$ are denoted as \emph{marked} nodes.
The output labels of $\Pi^{\lbadtree}$ are $\mathcal{V^\lbadtree} = \{\lerror, \bot\} \cup \{ (\lpointer,p) \mid p \in \{\lleft,\lright, \lparent,\lrch\} \}$. The constraints $\mathcal{C}^{\lbadtree}$ of $\Pi^{\lbadtree}$ are defined as follows.

\begin{myframe}{The constraints $\mathcal{C}^{\lbadtree}$}
\begin{enumerate}
	\item A node can always output $\bot$.\label{constr:pibad-bot}
	\item A node $u$ can output $\lerror$ only if $u$ does not satisfy $\mathcal{C}^{\ltreelike}$ or if $u$ is marked.\label{constr:noerror}
	\item Let $u$ be a node outputting $(\lpointer,p)$. It is required that node $u$ has an incident half-edge labeled $p$. Let $v$ be the node $f(u,p)$, that is, the node reachable from $u$ by following its half-edge labeled $p$. Then, either the output of $v$ is $\lerror$, or the output of $v$ is $(\lpointer,p')$. In the latter case, it must hold that $f(v,p) \neq u$, and the values of $p$ and $p'$ must satisfy the following:\label{constr:pointers}
	\begin{enumerate}
		\item If $p \in\{\lleft, \lright\} $ then $p' = p$;
		\item If $p = \lparent$ then $p' \in \{\lparent, \lleft,\lright\}$;
		\item If $p = \lrch$ then $p' \in \{\lrch, \lleft, \lright\}$.
	\end{enumerate}
\end{enumerate}
\end{myframe}

We now prove that, on properly labeled tree-like structures not containing any marked nodes, the only solution for  $\Pi^{\lbadtree}$ is the one where all nodes output $\bot$.
\begin{lemma}\label{lem:valid-sol-for-badtree}
	Let $G$ be a $(\{0,1\},\mathcal{E^\ltreelike})$-labeled graph where $\mathcal{C}^{\ltreelike}$ is satisfied on all nodes and all nodes are not marked. Then, the only valid solution for $\Pi^{\lbadtree}$ is the one assigning $\bot$ to all nodes.
\end{lemma}
\begin{proof}
	Constraint \ref{constr:pibad-bot} ensures that a solution where all nodes output $\bot$ is indeed possible, and constraint \ref{constr:noerror} ensures that there is no node in $G$ outputting the label $\lerror$. In the following, we prove that no node can output a pointer label, implying the claim.
	
	For a contradiction, assume that there is a node $u$ outputting $(\lpointer, p)$, for some $p \in \{\lleft,\lright, \lparent,\lrch\}$. Since no node outputs $\lerror$, then $v:=f(u, p)$ must output $(\lpointer, p')$ satisfying that $f(v,p')\neq u$. The same reasoning applies recursively on $v$. Hence, there must exist a path $(v_1, v_2, \ldots, v_k)$ where $v_1=u$, $v_2=v$, and each node $v_i$ outputs $(\lpointer,p_i)$, for some value $p_i$. Since no node outputs $\lerror$, then $v_k=v_i$ for some $i\le k-2$, forming a cycle. In the following, we prove that the constraints of $\Pi^{\lbadtree}$ ensure that cycles are not possible. 
	
	A necessary condition to properly label $(v_i, \ldots, v_k)$ with pointers is to use both labels $\lparent$ and $\lrch$, since any cycle on these nodes must start on some layer of the tree-like structure, then change layer, and eventually go back to the starting layer. However, the sequence of pointers allowed by constraint \ref{constr:pointers} must satisfy the regular expression $(\lparent^*|\lrch^*)(\lleft^*|\lright^*)$, which contradicts the fact that $\lparent$ and $\lrch$ are present at the same time.
\end{proof}

In \cite{congest-lcls}, it has been shown that, if the graph is not a valid tree-like structure, or if there is at least one marked node, then nodes can efficiently produce a proof of this fact.
\begin{lemma}[Lemma 6.10 in the arXiv version of \cite{congest-lcls}, rephrased]\label{lem:prove-tree-invalid}
	Let $G$ be a $(\{0,1\},\mathcal{E^\ltreelike})$-labeled graph where either $\mathcal{C}^{\ltreelike}$ is not satisfied on at least one node, or there is at least one marked node. Then, there exists a solution for $\Pi^{\lbadtree}$ where all nodes produce an output different from $\bot$. Moreover, such a solution can be computed in $O(\log n)$ rounds in the \local model.
\end{lemma}

\section{Grid structure}\label{sec:grid}
In this section, we first formally define what is a grid structure, then we provide some local constraints that, if satisfied on all nodes, guarantee that the graph is a grid (or something that locally looks like a grid everywhere), and finally we define some constraints that enforce a grid to be vertical.

\begin{definition}[Grid structure]
	We say that a graph $G$ is a \emph{grid stricture} of size $h\times w$ ($h,w>0$) if it is possible to assign coordinates $(x_u, y_u)$ to each node $u\in G$, where $0\le x_u < w$, $0 \le y_u < h$, such that there is an edge connecting two nodes $u,v\in G$ with coordinates $(x_u, y_u)$ and $(x_v, y_v)$ if and only if:
	\begin{itemize}[noitemsep]
		\item $x_u = x_v$ and $|y_u - y_v| = 1$, or
		\item $y_u = y_v$ and $|x_u - x_v| = 1$.
	\end{itemize} 
\end{definition}

\subsection{Local checkability of a grid structure}
We now define sets of labels $\mathcal{V^\lgrid}$ and $\mathcal{E^\lgrid}$, and a set of local constraints over such labels, satisfying that a graph $G$ can be $(\mathcal{V^\lgrid},\mathcal{E^\lgrid})$-labeled satisfying the constraints if and only of $G$ is a grid structure (with some exceptions that we will see later).
The grid structure that we use in this paper is the same exact one used in \cite{congest-lcls}, and hence we now report here the constraints as defined in \cite{congest-lcls}.

First of all, nodes are not labeled at all, or equivalently, all nodes have the same label $\bot$. Hence, we define $\mathcal{V^\lgrid} = \{\bot\}$. Then, the possible half-edge labels are $\mathcal{E^\lgrid} = \{\lup,\ldown, \lleft, \lright\}$ (the labels stand for ``up'', ``down'', ``left'', and ``right'', respectively). 
The set of local constraints $\mathcal{C}^{\lgrid}$ is defined as follows.

\begin{myframe}{The constraints $\mathcal{C}^{\lgrid}$ of \cite{congest-lcls}}
\begin{enumerate}
	\item For any two edges $e,e'$ incident to a node $u$, it must hold that $L_u(e)\neq L_u(e')$.\label{cons-grid:differentEdgeLabels}
	
	\item For each edge $e=\{u,v\}$, if $L_u(e)=\lleft$, then $L_v(e)=\lright$, and vice versa. \label{cons-grid:left-right}
	
	\item For each edge $e=\{u,v\}$, if $L_u(e)=\lup$, then $L_v(e)=\ldown$, and vice versa.\label{cons-grid:up-down}
	
	\item If a node $u$ has two incident edges labeled with $\lright$ and $\lup$ respectively, then it must hold that $f(u, \lright,\lup,\lleft,\ldown)=u$. \label{cons-grid:square}
	
	\item If $f(u,\lright)$ exists, then
	$u$ has an incident edge labeled with $\ldown$ (resp. $\lup$) if and only if  $f(u,\lright)$ has an incident edge labeled with $\ldown$ (resp. $\lup$). \label{cons-grid:down-propagates}
	
	\item If $f(u,\lup)$ exists, then
	$u$ has an incident edge labeled with $\lleft$ (resp. $\lright$) if and only if  $f(u,\lup)$ has an incident edge labeled with $\lleft$  (resp. $\lright$). \label{cons-grid:left-propagates}
\end{enumerate}
\end{myframe}
\noindent An example of valid grid structure labeled according to $\mathcal{C}^{\lgrid}$ is depicted in \Cref{fig:grid-labeled}.

In \cite{congest-lcls}, it has been shown that a graph $G$ can be $(\mathcal{V^\lgrid},\mathcal{E^\lgrid})$-labeled satisfying $\mathcal{C}^{\lgrid}$ if and only if it is a grid structure, with some exceptions. This result is captured by the following two lemmas.
\begin{lemma}[\cite{congest-lcls}]\label{lem:grid}
	Let $G$ be a graph that is $(\mathcal{V^\lgrid},\mathcal{E^\lgrid})$-labeled such that $\mathcal{C}^{\lgrid}$ is satisfied for all nodes in $G$. Moreover, assume that there exists at least one node that has no incident half-edge labeled $\ldown$ (or $\lup$), and that there exists at least one node that has no incident half-edge labeled $\lleft$ (or $\lright$). Then, $G$ is a grid structure.
\end{lemma}

\begin{lemma}[\cite{congest-lcls}]\label{lem:grid2}
	Each grid structure graph $G$ can be $(\mathcal{V^\lgrid},\mathcal{E^\lgrid})$-labeled such that the constraints $\mathcal{C}^{\lgrid}$ are satisfied at all nodes in $G$. 
\end{lemma}

\subsection{Additional constraints on the grid structure}
We now formally define what is a vertical grid structure.
\begin{definition}[Vertical grid structure]
	A $(\mathcal{V^\lgrid},\mathcal{E^\lgrid})$-labeled graph $G$ is a \emph{vertical grid structure} if and only if:
	\begin{itemize}[noitemsep]
		\item The graph $G$ is a grid structure and $\mathcal{C}^{\lgrid}$ is satisfied on all nodes;
		\item Let $w$ (resp.\ $h$) be the length of the paths induced by half-edges labeled $\lleft$ or $\lright$ (resp.\ $\lup$ or $\ldown$). Then, $h \ge w$.
	\end{itemize}
\end{definition}
We augment the $(\mathcal{V^\lgrid},\mathcal{E^\lgrid})$-labeling, and we define additional constraints, in order to obtain the local checkability of a vertical grid structure. The labeling for vertical grids is obtained by only changing the labels of the nodes. We define $\mathcal{V}^{\lvgrid}$ as $\{0,1\}$. The additional constraints are the following.
\begin{myframe}{The constraints to be added to $\mathcal{C}^{\lgrid}$ to obtain $\mathcal{C}^{\lvgrid}$}
\begin{enumerate}
	\item If a node $u$ is labeled $1$, then $f(u,\lup,\lright)$ and $f(u,\ldown,\lleft)$ are also labeled $1$, if they exist. \label{constr:diagonal}
	\item If a node $u$ is labeled $1$ and $f(u,\ldown) = \bot$, then $f(u,\lleft) = \bot$. \label{constr:vert1}
	\item If a node $u$ is labeled $1$ and $f(u,\lup) = \bot$, then $f(u,\lright) = \bot$.  \label{constr:vert2}
\end{enumerate}
\end{myframe}
\noindent We call $\mathcal{C}^{\lvgrid}$ the constraints obtained by adding the above additional constraints to  $\mathcal{C}^{\lgrid}$.
Note that, by labeling all nodes with $0$, the additional constraints in $\mathcal{C}^{\lvgrid}$ are trivially satisfied. Hence, we cannot claim that if $\mathcal{C}^{\lvgrid}$  is satisfied then the graph is a vertical grid. However, we now prove that a graph $G$ can be $(\mathcal{V^\lvgrid},\mathcal{E^\lgrid})$-labeled \emph{such that at least one node is labeled $1$} and satisfying $\mathcal{C}^{\lvgrid}$ if and only if it is a vertical grid structure.

\begin{lemma}\label{lem:vgrid}
	Let $G$ be a graph that is $(\mathcal{V^\lvgrid},\mathcal{E^\lgrid})$-labeled such that $\mathcal{C}^{\lvgrid}$ is satisfied on all nodes in $G$. Moreover, assume that there exists at least one node that has no incident half-edge labeled $\ldown$ (or $\lup$), that there exists at least one node that has no incident half-edge labeled $\lleft$ (or $\lright$), and that there exists at least one node labeled $1$. Then, $G$ is a vertical grid structure.
\end{lemma}
\begin{proof}
	By \Cref{lem:grid}, the graph is a grid structure, and hence it has some size $h \times w$. We prove that $h \ge w$. By assumption, there is at least one node labeled $1$, and let $(x,y)$ be its coordinates. By the additional constraint \ref{constr:diagonal} of $\mathcal{C}^{\lvgrid}$, for any integer $i$, each node with coordinates $(x+i,y+i)$, if it exists, is also labeled $1$. By constraint \ref{constr:vert1}, and by the fact that $G$ is a grid structure, there exists some $i$ such that the node $(x+i,y+i) = (0,y-x)$ exists, where $y-x \ge 0$. Similarly, by constraint \ref{constr:vert2}, and by the fact that $G$ is a grid structure, there exists some $i$ such that the node $(x+i,y+i) = (w-1,y + w - 1 - x)$ exists, where $y - x + w - 1 \le h -1$. By combining the two inequalities, we obtain $w \le h$, as required.
\end{proof}
\begin{lemma}\label{lem:vgrid2}
	Each vertical grid structure graph $G$ can be $(\mathcal{V^\lvgrid},\mathcal{E^\lgrid})$-labeled such that, at least one node is labeled $1$, and the constraints $\mathcal{C}^{\lvgrid}$ are satisfied at all nodes in $G$. 
\end{lemma}
\begin{proof}
	We label the nodes as follows. Initially, all nodes are labeled $0$. Then, we start from the node with coordinates $(0,0)$, i.e., the node not having any edge labeled $\ldown$ or $\lleft$, and we iteratively do the following. 
	\begin{itemize}[noitemsep]
		\item Label the current node $u$ with $1$.
		\item Move to $f(u,\lright,\lup)$ if it exists, otherwise stop.
	\end{itemize}
	Since the grid is vertical, this process necessarily ends on some node $u$ that does not have any edge labeled $\lright$, ensuring that the constraints are satisfied on all nodes.
\end{proof}

\section{Graph family \texorpdfstring{\boldmath$\mathcal{G}$}{G}}\label{sec:graph-family}

In this section, we first formally define the family $\mathcal{G}$ of hard instances, and then we define an LCL problem $\Pi^{\lbadgraph}$ that allows nodes to prove that some parts of the graph do not look like valid hard instances.

\begin{definition}[The graph family $\mathcal{G}$]
	A graph $G$ is in $\mathcal{G}$ if and only if $G$ can be constructed by the following process.
	Take an $h\times w$ grid structure $H$, where $h \ge w$, $h = 2^\ell$ for some integer $\ell \ge 0$, and $w > 0$. Take $w$ many tree-like structures $T_i$, where $i\in\{0,1,\dotsc, w-1\}$, of height $\ell$. Put each tree-like structure $T_i$ on top of the $i$th column of the grid structure $H$, that is, identify the node $u\in T_i$ with coordinates $(\ell-1,k_u)$ with the node $v\in H$ with coordinates $(i,k_u)$. 
\end{definition}
An example of a graph in the family $\mathcal{G}$ is depicted in \Cref{fig:hard-instance}.

\subsection{LCL problem \texorpdfstring{\boldmath$\Pi^{\lbadgraph}$}{Pi-badGraph}}

We now define a problem $\Pi^{\lbadgraph}$ that, informally, satisfies the following.
\begin{itemize}
	\item There are two possible types of output, either nodes give an empty output, or nodes produce a proof that the graph is not in the family.
	\item If the graph is in the family, the only valid solution is the one where all nodes produce an empty output.
	\item If the graph is not in the family, nodes can spend $O(\log n)$ time in the \local model to produce a proof of this fact. Moreover, the output can be constructed such that the subgraph induced by nodes producing an empty output satisfies that each connected component is a graph in the family.
\end{itemize}

\paragraph{Input.}
Each node is either marked to be both a grid node and a tree node, or only a tree node. Moreover, if it is a grid and tree node, it also receives an input label from $\mathcal{V}^{\lvgrid}$. That is, the possible inputs are $\mathcal{V^\lbadgraph} = \{ (\ltreenode)\} \cup \{(\ltreenode,\lgridnode,\ell) \mid  \ell \in \mathcal{V}^{\lvgrid}\}$. Each half-edge is either marked to be a grid edge, or a tree edge, or both. Moreover, grid half-edges also receive an input label from $\mathcal{E}^{\lgrid}$, and tree half-edges also receive an input label from $\mathcal{E}^{\ltreelike}$. Additionally, the input satisfies that all grid half-edges that receive the input $\lup$ (resp.\ $\ldown$) are also marked as tree edges labeled $\lright$ (resp.\ $\lleft$). Finally, all grid half-edges that receive the input $\lleft$ or $\lright$ are not marked as tree half-edges.
Hence, more formally, the possible half-edge labels are $\mathcal{E^\lbadgraph} = 
\{(\ltreeedge,\ell) \mid \ell \in \mathcal{E}^{\ltreelike}\} \cup 
\{(\lgridedge,\ell) \mid \ell \in  \{\lleft,\lright\}\} \cup
\{((\lgridedge,\ldown),(\ltreeedge,\lleft))\} \cup
\{((\lgridedge,\lup),(\ltreeedge,\lright))\}
$. An example of a ``good'' input labeling (i.e., that will force all nodes to produce an empty output) for $\Pi^{\lbadgraph}$ is depicted in \Cref{fig:hard-instance-labeled}.

\paragraph{Output.}
On a high level, there are two possible outputs.
\begin{itemize}
	\item Nodes can give an empty output. This is always an option, making the problem trivial.
	\item However, we allow nodes to prove that something is broken in their column, where a column of a node $u$ is defined as the connected component containing $u$ in the subgraph obtained by ignoring grid edges labeled $\lleft$ or $\lright$. A column is considered broken if edges are not consistently labeled, the tree-like structure is not valid, or there is some grid node in the column not satisfying the grid constraints.
\end{itemize}
We will prove that, if a graph belongs to $\mathcal{G}$, the only valid output is the empty output, while if a graph does not belong to $\mathcal{G}$, then there exists a valid output, computable in $O(\log n)$ rounds in the \local model, where the connected components of the subgraph induced by nodes producing an empty output satisfies that each connected component is a graph in $\mathcal{G}$.
More formally, the possible outputs are the following.
\begin{itemize}
	\item $\bot$: used by the nodes to produce an empty output.
	\item $\lerror$: used by the nodes to indicate that the edges are not consistently labeled (for example, if one half-edge of an edge is labeled as $\ltreeedge$, the other half-edge of the same edge must also be labeled $\ltreeedge$, or if a node is not marked as a grid node but it has incident grid edges).
	\item $\ltreeerror$: used by the nodes to indicate that the constraints $\mathcal{C}^{\ltreelike}$ are not satisfied.
	\item $\lgriderror$: used by the nodes to indicate that the constraints $\mathcal{C}^{\lvgrid}$ are not satisfied.
	\item $(\lcolumnerror,\ell)$, where $\ell \in \mathcal{V}^{\lbadtree}$: used by the nodes to indicate that the tree-like structure they belong to (i.e., their column in the grid) contains some error, which can be $\lerror$, $\lgriderror$, or $\ltreeerror$. For technical reasons, we identify the label $(\lcolumnerror,\bot)$ as the same label as $\bot$.
	\item $\lverterror$: used by the nodes to indicate that no node in the column they belong to has label $(\ltreenode,\lgridnode,1)$, i.e., no node in the column has input $1$ from $\mathcal{V}^{\lvgrid}$, meaning that the proof that the grid is vertical is missing.
\end{itemize}

\paragraph{Constraints.}
We now define the constraints of the problem $\Pi^{\lbadgraph}$.

We first define a function $t$ that, given the labeling of a half-edge label, extracts its \emph{type}.
\begin{definition}[Type of half-edge labels and types of edges]
	The type of a half-edge label is defined as the result of applying the function $t$ defined as follows:
  $t((\ltreeedge,\ell)) = \{\ltreeedge\}$, $t((\lgridedge,\ell)) = \{\lgridedge\}$, $t(((\ltreeedge,\ell),(\lgridedge,\ell'))) = \{\ltreeedge,\lgridedge\}$.
  
  Moreover, the type of an edge $e = \{u,v\}$ is defined as $\{\}$ if $t(L_u(e)) \neq t(L_v(e))$, and as $t(L_u(e))$ otherwise.
\end{definition}

Then, we define a function $\fvaluetree$ (resp. $\fvaluegrid$) that takes a label containing type $\ltreeedge$ (resp. $\lgridedge$) and returns the label associated with it.
\begin{definition}[Tree value of a half-edge label]
	The tree value of a half-edge label is defined as the result of applying the function $\fvaluetree$ defined as follows:
	\begin{align*}
		\fvaluetree((\ltreeedge,\ell)) &= \ell, \\
		\fvaluetree(((\ltreeedge,\ell),(\lgridedge,\ell'))) &= \ell,
	\end{align*}
	and it is undefined otherwise.
\end{definition}
\begin{definition}[Grid value of a half-edge label]
	The grid value of a half-edge label is defined as the result of applying the function $\fvaluegrid$ defined as follows:
	\begin{align*}
		\fvaluegrid((\lgridedge,\ell)) &= \ell, \\
		\fvaluegrid(((\ltreeedge,\ell),(\lgridedge,\ell'))) &= \ell',
	\end{align*}
	and it is undefined otherwise.
\end{definition}
We are now ready to define the constraints  $\mathcal{C}^{\lbadgraph}$ of $\Pi^{\lbadgraph}$.

\begin{myframe}{The constraints $\mathcal{C}^{\lbadgraph}$}
\begin{enumerate}
	\item The output $\bot$ is always allowed.
	\item A node $u$ can output $\lerror$ if it has an incident edge $e = \{u,v\}$ satisfying that 
	$t(L_u(e)) \neq t(L_v(e))$, or if its input is $(\ltreenode)$ and it has at least one incident half-edge with a type containing $\lgridedge$.\label{constr:badgraph:error}
	\item Consider the graph $G'$ induced by edges that have a type containing $\ltreeedge$. Consider the labeling of the half-edges of $G'$ obtained by mapping each half-edge $(u,e)$ (satisfying $u \in e$) into $\fvaluetree(L_u(e))$. A node can output $\ltreeerror$ if, in $G'$,  the constraints $\mathcal{C}^{\ltreelike}$ are not satisfied.\label{constr:badgraph:treeerror}
	\item Consider the graph $G'$ induced by edges that have a type containing $\lgridedge$. Consider the labeling of the half-edges of $G'$ obtained by mapping each half-edge $(u,e)$ (satisfying $u \in e$) into $\fvaluegrid(L_u(e))$. A node can output $\lgriderror$ if its input is not $(\ltreenode)$ and, in $G'$,  the constraints $\mathcal{C}^{\lgrid}$ are not satisfied.\label{constr:badgraph:griderror}
	\item Consider the graph $G'$ induced by edges that have a type containing $\ltreeedge$. Consider the following labeling of $G'$: \label{constr:badgraph:badtree}
	\begin{itemize}
		\item The input labeling of the half-edges of $G'$ is obtained by mapping each half-edge $(u,e)$ (satisfying $u \in e$) into $\fvaluetree(L_u(e))$.
		\item The output labeling of the nodes of $G'$ is obtained by mapping each node $u$ labeled $(\lcolumnerror,\ell)$ into $\ell$, each node $u$ labeled $\lerror$, $\ltreeerror$, or $\lgriderror$ into $\lerror$, and each other node into $\bot$.
		\item The input labeling of the nodes is obtained by labeling $1$ nodes that have an output in $\{\lerror, \ltreeerror, \lgriderror\}$, and $0$ all other nodes.
	\end{itemize}
	Then, on $G'$, the constraints of $\mathcal{C}^{\lbadtree}$ must be satisfied.
	\item Consider the graph $G'$ induced by edges that have a type containing $\ltreeedge$. If a node has output $\lverterror$, then all its neighbors in $G'$ must also have $\lverterror$ as output. Moreover, if a node has output $\lverterror$, then it must not have input $(\ltreenode,\lgridnode,1)$. \label{constr:badgraph:verterror}
\end{enumerate}
\end{myframe}

\paragraph{Properties of the problem.}
We now prove some useful properties about the problem $\Pi^{\lbadgraph}$.
\begin{lemma}\label{lem:badgraph-bot-everywhere}
	Let $G$ be a graph in $\mathcal{G}$. There exists a $(\mathcal{V^\lbadgraph},\mathcal{E^\lbadgraph})$-labeling of $G$ satisfying that the only valid output labeling for $\Pi^{\lbadgraph}$ is the one assigning $\bot$ to all nodes.
\end{lemma}
\begin{proof}
	Recall that any $G \in \mathcal{G}$ is obtained by starting from a vertical grid structure and then connecting a tree-like structure on top of each column. 
	For each node $u$ that is part of the grid structure, let $\ell_g(u)$ be the label of node $u$ obtained by \Cref{lem:vgrid2}. For each half-edge $(u,e)$ that is part of the grid structure, let $\ell_g(u,e)$ be the label of the half-edge $(u,e)$ obtained by \Cref{lem:vgrid2}. For each half-edge that is part of the tree-like structure, let $\ell_t(u,e)$ be the label of the half-edge $(u,e)$ given by \Cref{lem:treelike2}. For each node $u$ of $G$, if it is only part of a tree-like structure, we assign the label $(\ltreenode)$, while if $u$ is part of both a tree-like structure and the grid, we assign the label $(\ltreenode,\lgridnode,\ell_g(u))$.
	Then, to each half-edge $(u,e)$ that is only part of the tree-like structure we assign the label $(\ltreeedge,\ell_t(u,e))$, to each half-edge $(u,e)$ that is only part of the grid structure we assign the label $(\lgridedge,\ell_g(u,e))$, while to all other half-edges we assign the label $((\lgridedge,\ell_g(u,e)),(\ltreeedge,\ell_t(u,e)))$.
	
	By construction of the input labeling, the outputs $\lerror$, $\ltreeerror$, and $\lgriderror$ are not allowed. Moreover, since by \Cref{lem:treelike2} at least one node is labeled $(\ltreenode,\lgridnode,1)$, the output $\lverterror$ is not allowed on that node, and by a propagation argument no node belonging to the same tree-like structure can use the output $\lverterror$. 
	
	We thus get that the only possible outputs are $\bot$ or $(\lcolumnerror,\ell)$ for some $\ell$. We prove that $\ell$ must be $\bot$, implying the claim.
	The fact that no node is allowed to output $\lerror$, $\ltreeerror$, $\lgriderror$, or $\lverterror$, implies, by constraint \ref{constr:badgraph:badtree} of $\mathcal{C}^{\lbadgraph}$, that the input labeling for $\Pi^{\lbadtree}$ satisfies that all nodes receive $0$ as input. Since the tree-like structure is valid, by \Cref{lem:valid-sol-for-badtree}, the only way to satisfy $\mathcal{C}^{\lbadtree}$ is for all nodes to output $\bot$ as output for $\Pi^{\lbadtree}$.
\end{proof}
\begin{lemma}\label{lem:solve-badgraph}
	For any $(\mathcal{V^\lbadgraph},\mathcal{E^\lbadgraph})$-labeled graph $G$, there exists a solution for $\Pi^{\lbadgraph}$ satisfying the following.
	\begin{itemize}[noitemsep]
		\item The solution can be computed in $O(\log n)$ rounds in the \local model.
		\item The graph induced by nodes that output $\bot$ satisfies that each connected component is a graph in~$\mathcal{G}$.
	\end{itemize}
\end{lemma}
\begin{proof}
	We present an algorithm that satisfies the requirements of the lemma.
	At first, each node $u$ spends $O(1)$ rounds to check whether there is some local inconsistency in the structure of the graph. In particular, each node $u$ first checks whether constraint \ref{constr:badgraph:error} applies, and in that case it outputs $\lerror$. Then, it checks whether constraint \ref{constr:badgraph:treeerror} applies, and in that case it outputs $\ltreeerror$. Then, it checks whether constraint \ref{constr:badgraph:griderror} applies, and in that case it outputs $\lgriderror$. 
	
	Then, nodes construct the input for $\Pi^{\lbadtree}$ by marking themselves $1$ if they have output $\lerror$, $\ltreeerror$, or $\lgriderror$, and $0$ otherwise. In each connected component obtained by ignoring grid edges labeled $\lleft$ or $\lright$, by \Cref{lem:prove-tree-invalid}, nodes can then spend $O(\log n)$ rounds to produce a solution of $\Pi^{\lbadtree}$ such that, if the tree-like structure they belong to is invalid or some node is marked, then no node of the connected component outputs $\bot$, while if the tree-like structure they belong to is valid and no node is marked, then all nodes output $\bot$. A node that obtained output $\ell$ outputs $(\lcolumnerror,\ell)$. Recall that $(\lcolumnerror,\bot) = \bot$. Note that this output satisfies the requirements of constraint \ref{constr:badgraph:badtree}.
	
	Finally, each node $u$ that output $\bot$ in the previous step, can spend $O(\log n)$ rounds to gather the whole tree-like structure it belongs to (since a valid tree-like structure has diameter $O(\log n)$). If $u$ does not see any node labeled $(\ltreenode,\lgridnode,1)$, it changes its output to $\lverterror$. Since this operation is done consistently by all nodes, constraint \ref{constr:badgraph:verterror} is satisfied. 
	
	We obtained an algorithm that produces a correct output for $\Pi^{\lbadgraph}$ in $O(\log n)$ rounds. We now prove that the output satisfies the second requirement of the lemma.
	Summarizing the above, we get that if a node $u$ outputs $\bot$ then:
	\begin{itemize}[noitemsep]
		\item the grid constraints are locally satisfied;
		\item the tree-like structure containing $u$ is valid;
		\item by a propagation argument, the tree-like structure is correctly connected to the whole column of the grid;
		\item all the nodes in such a column also output $\bot$;
		\item the column of the grid contains at least one node labeled $1$.
	\end{itemize} 
	Hence, the subgraph $G'$ induced by nodes that output $\bot$ is composed of whole columns of the grid with tree-like structures correctly attached. Moreover, if two such columns are connected, then the additional constraints of $\mathcal{C}^{\lvgrid}$ are also satisfied.
	
	Consider a connected component $G''$ in $G'$. Since the tree-like structures are valid and properly connected to the columns, no column can \emph{wrap around}, i.e., for each column there is at least one grid node not having any half-edge labeled $\ldown$ and at least one grid node not having any half-edge labeled $\lup$. We now prove that, in the connected component $G''$, there is at least one node not having any incident half-edge labeled $\lleft$ (i.e., the grid does not wrap around horizontally). Consider an arbitrary column $c$ of $G''$, and let $u$ be an arbitrary node labeled $1$ in $c$ (such a node, by assumption, exists), and let $i$ be the vertical coordinate of $u$ in $c$. Consider the column $c'$ obtained by starting from $u$ and moving left for $i$ steps. By the definition of $\mathcal{C}^{\lvgrid}$, the column $c'$ must satisfy that the node with vertical coordinate $0$ has input $1$, which, again by the definition of $\mathcal{C}^{\lvgrid}$, implies that the node does not have any half-edge labeled $\lleft$, as desired.
	
	Recall that, when nodes of $G''$ checked whether the constraints of $\mathcal{C}^{\lvgrid}$ were satisfied, they did it on $G$. However, it is easy to see that the constraints $\mathcal{C}^{\lvgrid}$ satisfy a special property: if a node $u$ satisfies $\mathcal{C}^{\lvgrid}$, and an entire column is removed from the left or the right of $u$, then the constraints are still satisfied on $u$. This implies that, even if we restrict to $G''$, the constraints $\mathcal{C}^{\lvgrid}$ are still satisfied on all nodes of $G''$. Hence, we now have all the requirements for applying \Cref{lem:vgrid} and proving that the nodes of $G''$ labeled as grid nodes form a vertical grid structure, which implies our claim.
\end{proof}

\section{LCL problem \texorpdfstring{\boldmath$\Pi$}{Pi}}\label{sec:problem-pi}
On a high level, the problem $\Pi$ will be defined in such a way that it requires nodes to solve $\Pi^{\lbadgraph}$, and then, additionally, nodes that output $\bot$ for $\Pi^{\lbadgraph}$ are required to solve an additional problem. Such a problem requires non-local coordination.

\paragraph{Input.}
The set of input labels of the nodes is $\mathcal{V}^\Pi=\mathcal{V^\lbadgraph}\times \{0,1\}$, that is, nodes receive an input of the LCL $\Pi^\lbadgraph$ and an additional bit. The set of input labels of half-edges is $\mathcal{E}^\Pi=\mathcal{E^\lbadgraph}$. 

\paragraph{Output.}
The set of output labels of the nodes contains the following labels.
\begin{itemize}
	\item All labels in $\mathcal{V^\lbadgraph}\setminus\{\bot\}$. A valid solution for $\Pi$ will be to solve $\Pi^{\lbadgraph}$ without using the label $\bot$. In this case, nothing additional will be required.
	\item All labels in $\{0,1\} \times \{\lyes,\lno\}$. These outputs will be used by grid nodes.
	\item The labels $\lyes$ and $\lno$, which will be used by nodes that are not in the grid.
\end{itemize} 

\paragraph{Constraints.}
We now define the node constraints $\mathcal{C}^\Pi$ of $\Pi$.
\begin{myframe}{The constraints $\mathcal{C}^{\Pi}$}
\begin{enumerate}
	\item Consider the output labeling given by mapping all labels not in $\mathcal{V^\lbadgraph}$ to $\bot$. The constraints of $\Pi^\mathcal{\lbadgraph}$ must be satisfied.\label{constr:pi:bad}
	
	\item If a node is labeled $(\ltreenode,\lgridnode,\ell)$ for some $\ell$ and its output is not in $\mathcal{V^\lbadgraph}$, then its output must be in  $\{0,1\} \times \{\lyes,\lno\}$. \label{constr:pi:grid}
	
	\item Let $u$ be a node with output $(b,x) \in \{0,1\} \times \{\lyes,\lno\}$. Let $v$ be $f(u,(\lgridedge,\lright))$. If $v$ exists, the output of $v$ must be either a label in $\mathcal{V^\lbadgraph}$ or $(b,x')$ for some $x'$. \label{constr:pi:grid-same-output}
	
	\item Let $u$ be a node with input $(\ell,b_\mathrm{in})$ for some $\ell$, output $(b_\mathrm{out},x) \in \{0,1\} \times \{\lyes,\lno\}$, and such that $f(u,(\lgridedge,\lright)) = \bot$ (i.e., $u$ does not have a right neighbor in the grid). It must hold that $x = \lyes$ if and only if $b_\mathrm{in} = b_\mathrm{out}$. \label{constr:pi:basecase}
	
	\item If a node is labeled $(\ltreenode)$ and its output is not in $\mathcal{V^\lbadgraph}$, then its output must be in $\{\lyes, \lno\}$. \label{constr:pi:tree}

	\item Let $u$ be a node with output $x_u \in \{\lyes,\lno\}$, i.e., it is a node that does not belong to the grid, but it belongs to a tree-like structure and did not give an output from $\mathcal{V^\lbadgraph}$. Let $v$ be the node $f(u,(\ltreelike,\llch))$, and let $z$ be the node $f(u,(\ltreelike,\lrch))$. It must hold that the output of $v$ is either $x_v \in \{\lyes,\lno\}$, or $(b_v,x_v) \in \{0,1\} \times \{\lyes,\lno\}$. Similarly, the output of $z$ is either $x_z \in \{\lyes,\lno\}$, or $(b_z,x_z) \in \{0,1\} \times \{\lyes,\lno\}$. It must hold that $x_u = \lyes$ if and only if at least one of $x_v$ or $x_z$ is $\lyes$. \label{constr:pi:or}
	
	\item Let $u$ be a node with output $x_u \in \{\lyes,\lno\}$ and such that $f(u,(\ltreelike,\lparent)) = \bot$. Then, $x_u$ must be $\lyes$. \label{constr:pi:yes}
\end{enumerate}
\end{myframe}
\noindent An example of valid output is shown in \Cref{fig:solution}.

\paragraph{\boldmath Properties of $\Pi$.}
We now characterize what are the valid solutions for $\Pi$, when we consider hard instances.
\begin{lemma}\label{lem:pi-good-sol}
	Let $G \in \mathcal{G}$. It is possible to label $G$ such that the only valid solutions for $\Pi$ satisfy the following.
	\begin{itemize}
		\item For each row of the grid structure, all nodes in that row output the same bit.
		\item There exists at least one row satisfying that the output bit given by the nodes is the same as the one provided as input to the right-most node in the row.
	\end{itemize}
\end{lemma}
\begin{proof}
	By \Cref{lem:badgraph-bot-everywhere}, there exists a $(\mathcal{V^\lbadgraph},\mathcal{E^\lbadgraph})$-labeling of $G$ where the only output labeling satisfying  $\mathcal{C}^{\lbadgraph}$ is the one assigning $\bot$ to all nodes. Consider this labeling as input for the problem $\Pi$. Nodes cannot use any output from $\mathcal{V^\lbadgraph}\setminus\{\bot\}$.
	
	By constraint \ref{constr:pi:grid} of $\mathcal{C}^\Pi$, each node of the grid must output a pair $\{0,1\} \times \{\lyes,\lno\}$, and by constraint \ref{constr:pi:grid-same-output} of $\mathcal{C}^\Pi$ for each row of the grid it must hold that the first element of the pairs given by the nodes must be the same, i.e., all $0$ or all $1$, which implies the first property of the lemma.
	
	By constraint \ref{constr:pi:basecase}, a grid node of the right-most column is allowed to output $(b,\lyes)$ for some $b$ only if $b$ matches the bit received as input. 
	Then, by constraints \ref{constr:pi:tree} and \ref{constr:pi:or}, a node of a tree-like structure is allowed to output $\lyes$ only if at least one of its children outputs also $\lyes$. Since, by 
	constraint \ref{constr:pi:yes}, the root of the tree-like structure in the right-most column must output $\lyes$, we get that at least one node of the right-most column outputs $(b,\lyes)$ for some $b$, which implies the second property of the lemma.
\end{proof}

\section{Complexity in the \local model}\label{sec:pi-complexity}

In this section we will show that our problem $\Pi$ is indeed hard in the randomized \local model with private randomness, but easy with shared randomness.

\subsection{Complexity with private randomness}

\begin{theorem}\label{th:private-rand}
	In the \local model with private randomness, solving the problem $\Pi$ with success probability at least $1 - 1/n$ requires $\Omega(\sqrt{n})$ rounds.
\end{theorem}
\begin{proof}
We consider the subfamily of graphs in $\mathcal{G}$ satisfying that the dimensions $w$ and $h$ of the grid satisfy $w = h$, and we apply \Cref{lem:pi-good-sol}. 

	Assume for a contradiction that there exists an algorithm $\mathcal{A}$ that solves $\Pi$ in $o(\sqrt{n})$ rounds. 
 
	Let $N$ be such that, for any $n\ge N$, the time complexity of $\mathcal{A}$
	in graphs of size $n$ is at most $w/3$. Consider a row, its left-most node $u$ and its right-most node $v$. By the assumption on the runtime of the algorithm, their outputs are independent. 
	Let $p_u$ (resp.\ $p_v$) be the probability that $u$ (resp.\ $v$) outputs $0$. It must hold that $p_u (1 - p_v) < 1/n$ and that $(1- p_u)p_v < 1/n$, since otherwise, the algorithm would produce a row that does not have the same bit everywhere (which contradicts \Cref{lem:pi-good-sol}) with too large probability. This implies that either both $p_u$ and $p_v$ are at most $2/n$ or that they are both at least $1 - 2/n$.
	
	We now restrict to instances where, for each row with left-most node $u$ and right-most node $v$, the input of $v$ is $0$ if $p_u \le 2/n$, and $1$ otherwise. On these instances, the success probability of the algorithm is upper bounded by the probability that at least one left-most node picks its least probable output, which, since the number of rows is upper bounded by $O(\sqrt{n})$, happens with probability at most $c / \sqrt{n}$ for some constant $c$. This probability, for large enough $n$, is strictly smaller than $1 - 1/n$, implying that the algorithm fails with too large probability.
\end{proof}

\subsection{Complexity with shared randomness}
\begin{theorem}\label{thm:ub-shared-rand}
	In the \local model with shared randomness, the problem $\Pi$ can be solved in $O(\log n)$ rounds, with success probability $1 - 1/n^c$ for any constant $c$.
\end{theorem}
\begin{proof}
	By \Cref{lem:solve-badgraph}, nodes can spend $O(\log n)$ rounds to produce a solution for $\Pi^{\lbadgraph}$ where the graph induced by nodes that output $\bot$ satisfies that each connected component is a graph in $\mathcal{G}$. Observe that such labeling satisfies constraint \ref{constr:pi:bad} of $\mathcal{C}^\Pi$.
	
	Nodes that have an output different from $\bot$ immediately terminate. Then, each node $u$ does the following. Node $u$ spends $O(\log \log n)$ rounds to check if the height of the grid structure is at most $c \log n$. In this case, the width of the grid is also guaranteed to be at most $c \log n$. This implies that nodes are in a small connected component that is a valid graph of the family. In this case, each node $u$ can spend $O(\log n)$ rounds to know the bit $b_v$ given as input to the right-most node $v$ in the row of $u$. In this case, let $b_u = b_v$.
	
	Otherwise, i.e., the height is strictly larger than $c \log n$, if $u$ is a grid-node, it spends $O(\log n)$ rounds to compute its position $i$ in its column. Then, $u$ sets $b_u$ as the $i$th shared random bit. Note that this implies that nodes in the same row pick the same random bit.

	Then, if $u$ is not in the last column, it outputs $(b_u,\lyes)$, while if $u$ is in the last column it outputs $(b_u,x_u)$, where $x_u = \lyes$ if $x_u$ is equal to the bit given as input to $u$, and $x_u = \lno$ otherwise.
	Finally, nodes in the tree-like structures output $\lyes$ if at least one of their children has $\lyes$, and $\lno$ otherwise.
	
	By construction of the output, all constraints \ref{constr:pi:grid}--\ref{constr:pi:or} are clearly satisfied. If the grid structure has height at most $c \log n$, then all nodes in the trees output $\lyes$ and constraint \ref{constr:pi:yes} is also satisfied. If the grid structure has height strictly larger than $c \log n$, since for each right-most node $v$ it holds that the randomly picked bit is the same as its input bit with probability $1/2$, constraint \ref{constr:pi:yes} is satisfied with probability at least $1 - 1 / 2^{c\log n} = 1 - 1 / n^c$, as required.
\end{proof}

\section{Complexity in other models}
\label{sec:other-models}

In this section we explore the complexity of problem $\Pi$ in other models, beyond the usual \local model.

\subsection{S\local model with private randomness}\label{sec:slocal}

In the \slocal model~\cite{ghaffari17_complexity_stoc}, nodes are processed
sequentially according to an ordering $\sigma = v_1,\dots,v_n$. 
The order $\sigma$ is controlled by an adversary; that is, an \slocal algorithm
must work for any such sequence $\sigma$.
Let $T$ be given as a function of $n$.
When processing node $v_i$, an algorithm with time complexity $T$ has access to
the neighborhood at distance $T$ of $v_i$, including whatever information the
nodes there may have in their memory (and thus also their output).
In the randomized version of the model (with private randomness), the algorithm
is given access to a (read-once) infinite sequence of random bits.

\begin{theorem}\label{thm:lb-slocal}
	In the \slocal model with private randomness, solving the problem $\Pi$
	requires time $\Omega(\sqrt{n})$.
\end{theorem}
\begin{proof}
We consider the subfamily of graphs in $\mathcal{G}$ with grid dimensions $w$
and $h$ where $w = h$, and we apply \Cref{lem:pi-good-sol}. 

	Assume towards a contradiction that there exists an algorithm $\mathcal{A}$
	that solves $\Pi$ in $o(\sqrt{n})$ rounds. The difference compared to the
	proof of Theorem~\ref{th:private-rand} is that we need to provide a suitable
	ordering under which the nodes are processed. Having done so, the proof
	remains the same.

	Let $N$ be such that, for any $n\ge N$, the time complexity $T$ of
	$\mathcal{A}$ in graphs of size $n$ is at most $w/3$. Consider the nodes
	in the left- and right-most columns of the grid. Observe that, if we select
	any pair of nodes in the two columns, their $T$-radius neighborhoods do not intersect.

	Let us now fix the processing order such that all the nodes of the left-most
	column come first, followed by the nodes in the right-most column (from top to
	bottom), and finally by the remaining nodes in arbitrary order.

	Consider a row where its left-most node is $u$ and its right-most one $v$. By
	construction, the outputs of $u$ and $v$ must be independent. 
	Let $p_u$ (resp., $p_v$) be the probability that $u$ (resp., $v$) outputs $0$.
	Observe that $p_u (1 - p_v) < 1/n$ and $(1- p_u)p_v < 1/n$; otherwise, the
	algorithm would produce a row that does not have the same bit everywhere
	(contradicting \Cref{lem:pi-good-sol}) with too large probability. This
	implies that either both $p_u$ and $p_v$ are at most $2/n$ or that they are
	both at least $1 - 2/n$.
	
	We now restrict to instances where, for each row with left-most node $u$ and
	right-most node $v$, the input of $v$ is $0$ if $p_u \le 2/n$, and $1$
	otherwise. On these instances, the success probability of the algorithm is
	upper-bounded by the probability that at least one left-most node picks its
	least probable output. Since the number of rows is $O(\sqrt{n})$, by the union
	bound this happens with probability at most $c / \sqrt{n}$ for some constant
	$c$. 
	For large enough $n$, this is strictly smaller than $1 - 1/n$, implying that
	the algorithm fails with too large probability.
\end{proof}

\subsection{\Detolcl model}\label{sec:detolcl}

The \onlinelocal model~\cite{akbari23_locality_icalp} is slightly more
powerful than \slocal. Again we have an (adversarial) sequence $\sigma =
v_1,\dots,v_n$ in which the nodes are processed. Upon node $v_i$ being revealed,
the algorithm is given knowledge of all nodes (including their inputs) and their
connections in the $T$-radius neighborhood of $v_i$ (in addition to everything
the algorithm has seen so far).

Here we restrict ourselves to the deterministic variant of \onlinelocal.

\begin{theorem}\label{thm:lb-detolcl}
	In the \detolcl model, solving the problem $\Pi$ requires $\Omega(\sqrt{n})$
	rounds.
\end{theorem}
\begin{proof}
	As in the proof of \cref{thm:lb-slocal}, we consider the subfamily of graphs
	in $\mathcal{G}$ with grid dimensions $w$ and $h$ where $w = h$, and we apply
	\Cref{lem:pi-good-sol}. 
	Assume towards a contradiction that there exists an algorithm $\mathcal{A}$
	that solves $\Pi$ in $o(\sqrt{n})$ rounds. 
	Let $N$ be such that, for any $n\ge N$, the time complexity of $\mathcal{A}$
	in graphs of size $n$ is at most $w/3$.
	
	Consider the processing order where all nodes in the left-most column come
	first in the sequence, followed by all other nodes in arbitrary order.
	As the algorithm is deterministic, it must produce some output on each node,
	in particular without any knowledge of the input to the right-most column.
	Hence we can easily pick an adversarial input:
	For each node $u$ on the left-most column, letting $o_u$ be its output and $v$
	the corresponding right-most node in the same row as $u$, we set $1-o_u$ as
	the input to $v$.
	Since the algorithm cannot modify its outputs, it does not produce a valid
	solution to $\Pi$.
\end{proof}

\subsection{\Boundep}\label{sec:boundep}

The \boundep~\cite{akbari24_online_arxiv} (with locality $T$)
encompasses any algorithm where the output distributions of nodes that are at
distance more than $T$ from one other are independent.
Equivalently, given two graphs $G_1$ and $G_2$ (with inputs) with the same
number of nodes, if any subgraph $H_1$ of $G_1$ is isomorphic to a subgraph
$H_2$ of $G_2$, then the algorithm must have identical marginal distributions on
$H_1$ and $H_2$.

\begin{theorem}\label{thm:lb-boundep}
	In the \boundep, solving the problem $\Pi$ requires locality $\Omega(\sqrt{n})$.
\end{theorem}
\begin{proof}
	The proof of Theorem~\ref{th:private-rand} works directly here since, for each
	row, the output distributions of the left- and right-most nodes are
	independent (as their neighborhoods do not intersect).
\end{proof}

\ifanon\else
\section*{Acknowledgements}

This work was started in the \emph{Research Workshop on Distributed Algorithms} (RW-DIST 2024) in L'Aquila, Italy. We would like to thank all workshop participants for discussions, and in particular Henrik Lievonen and Amirreza Akbari for helping us with this research project. This work has been partially funded by MUR (Italy) Department of Excellence 2023--2027, the PNRR MIUR research project GAMING ``Graph Algorithms and MinINg for Green agents'' (PE0000013, CUP D13C24000430001). Augusto Modanese is supported by the Helsinki Institute for Information Technology (HIIT).

\fi % end ifanon

\printbibliography

\newpage
\appendix

\section{Monte Carlo algorithms require some knowledge of \texorpdfstring{\boldmath$n$}{n}}
\label{appx:monte-carlo}

In this appendix, we show that any truly (i.e., non-deterministic) Monte Carlo
algorithm (regardless of whether it uses shared or private randomness) for any
component-wise checkable problem requires some form of knowledge of the number
of nodes $n$, regardless of its locality.
Note component-wise checkable problems are a much broader class than LCLs.

\begin{definition}[Component-wise checkability]
	A labeling problem $\Pi$ is \emph{component-wise checkable} if the following
	holds: 
	For every labeled graph $G$ and every connected component $H$ of $G$, $G$
	satisfies $\Pi$ if and only if its labeled subgraphs $H$ and $G - H$ (seen as
	two separate graphs) both satisfy $\Pi$.
\end{definition}

An example of a problem that does not qualify as such is non-component-wise
leader election, that is, the labeled graph contains exactly one node marked as
the leader.
(Of course, component-wise leader election is a component-wise checkable
problem.)

\begin{definition}[Monte Carlo algorithm]\label{def:monte-carlo}
	A randomized distributed algorithm $\mathcal{A}$ is said to be a \emph{Monte
	Carlo} algorithm if, for any constant $c > 0$, there is $n_0 > 0$ such that,
	if $\mathcal{A}$ is ran on a graph $G$ with $n \ge n_0$ nodes, then the
	success probability of $\mathcal{A}$ is at least $1 - n^{-c}$.
	The algorithm $\mathcal{A}$ is \emph{error-free} if its success probability is
	$1$.
\end{definition}

In particular, error-free algorithms can be trivially derandomized.

\begin{theorem}
	Let $\Pi$ be component-wise checkable, and let $\mathcal{A}$ be a Monte Carlo
	algorithm that solves $\Pi$ (with any locality) and is given no knowledge of
	$n$ whatsoever.
	Then $\mathcal{A}$ is error-free.
\end{theorem}

\begin{proof}
	Suppose there is a graph $G$ for which $\mathcal{A}$ is not error-free, that
	is, $\mathcal{A}$ fails on $G$ with probability $p > 0$.
	Let $c > 0$ be fixed, and let $n_0$ be as in \cref{def:monte-carlo}.
	Consider the graph $G' = G \cup H$ where $H$ is disconnected from $G$ and
	contains $N > 1/p^{1/c}$ nodes.
	Since $\mathcal{A}$ is given no knowledge of the number of nodes of $G'$, its
	distribution of outputs on $G$ as a component of $G'$ is identical to the
	distribution obtained when running $\mathcal{A}$ on $G$ alone.
	(Failing on the component $G$ might be correlated with failing on $H$, but
	this is immaterial.)
	In particular the failure probability of $\mathcal{A}$ on $G'$ is at least $p
	> 1/N^c$, and thus its success probability is strictly smaller than $1 -
	N^{-c}$, contradicting \cref{def:monte-carlo}.
	It follows that $p = 0$.
\end{proof}

\end{document}